\documentclass[10pt,journal,compsoc]{IEEEtran}
\usepackage{amsmath}
\usepackage{amsfonts}
\usepackage{amsthm}
\usepackage{makecell}
\usepackage{algorithmic}
\usepackage{graphicx}
\usepackage{textcomp}
\usepackage[dvipsnames,table]{xcolor}
\usepackage{tikz}
\usepackage{multicol}
\usepackage{subcaption}
\usepackage{pgfplots}
\pgfplotsset{compat=1.18}
\usepackage{tcolorbox}
\usepackage{url}
\usepackage{hyperref}
\usepackage[ruled,linesnumbered,vlined,nofillcomment]{algorithm2e} 
\usepackage{bm}
\usepackage{cleveref}
\usepackage{multirow}
\usetikzlibrary{shapes.geometric}

\newtheorem{theorem}{Theorem}
\newtheorem{assumption}{Assumption}

\usepackage{hhline}
\usepackage{enumitem}

\newcommand\colorrow{\rowcolor{gray!25}}
\newcommand\highlightrow{\rowcolor{green!10}}

\newcommand{\MethodName}{\texttt{Ampere}}
\newif\ifrevhighlight
\revhighlightfalse
\newcommand{\rev}[1]{\ifrevhighlight\textcolor{blue}{#1}\else#1\fi}

\newcommand{\cir}[1]{\tikz[baseline]{%
    \node[anchor=base, fill=black, draw, circle, inner sep=0, minimum width=1.1em]{#1};}}

\ifCLASSOPTIONcompsoc
  \usepackage[nocompress]{cite}
\else
  \usepackage{cite}
\fi

\begin{document}

\title{\texttt{Ampere}: Communication-Efficient and High-Accuracy Split Federated Learning}

\author{
    Zihan~Zhang,
    Leon~Wong,
    and~Blesson~Varghese
    \IEEEcompsocitemizethanks{
        \IEEEcompsocthanksitem Z. Zhang and B. Varghese are with the School of Computer Science, University of St Andrews, UK. Corresponding author: bv6@st-andrews.ac.uk
        \IEEEcompsocthanksitem W. Leon is with Rakuten Mobile, Inc., Japan.
    }
}


\IEEEtitleabstractindextext{%
\begin{abstract}
A Federated Learning (FL) system collaboratively trains neural networks across devices and a server but is limited by significant on-device computation costs. 
Split Federated Learning (SFL) systems mitigate this by offloading a block of layers of the network from the device to a server. 
However, in doing so, it introduces large communication overheads due to frequent exchanges of intermediate activations and gradients between devices and the server and reduces model accuracy for non-IID data. 
We propose \MethodName, a novel collaborative training system that simultaneously minimizes on-device computation and device-server communication while improving model accuracy. 
Unlike SFL, which uses a global loss by iterative end-to-end training, \MethodName\ develops unidirectional inter-block training to sequentially train the device and server blocks with a local loss, eliminating the transfer of gradients. 
A lightweight auxiliary network generation method decouples training between the device and server, reducing frequent intermediate exchanges to a single transfer, which significantly reduces the communication overhead. 
\MethodName\ mitigates the impact of data heterogeneity by consolidating activations generated by the trained device block to train the server block, in contrast to SFL, which trains on device-specific, non-IID activations.
Extensive experiments on multiple CNNs and Transformers show that, compared to state-of-the-art SFL baseline systems, \MethodName\ (i) improves model accuracy by up to \rev{11.70 percentage points} while training up to \rev{18.6$\times$} faster, (ii) incurs up to \rev{911$\times$} lower device-server communication overhead and up to \rev{14.5$\times$} lower on-device computation, and (iii) reduces standard deviation of accuracy by \rev{71.13\%} for various non-IID degrees highlighting superior performance when faced with heterogeneous data. \rev{\MethodName\ is available from \url{https://github.com/blessonvar/Ampere}.}
\end{abstract}

\begin{IEEEkeywords}
 Federated learning, split federated learning, communication cost, on-device computation, data heterogeneity.
\end{IEEEkeywords}}

\maketitle

\IEEEdisplaynontitleabstractindextext

%
\IEEEpeerreviewmaketitle

\IEEEraisesectionheading{\section{Introduction}\label{sec:intro}}
A federated learning (FL)~\cite{fedavg, DBLP:journals/corr/KonecnyMR15,DBLP:journals/corr/KonecnyMRR16,DBLP:journals/corr/KonecnyMYRSB16} system enables multiple devices to collaboratively train a deep neural network (DNN), such as convolutional neural networks (CNN) and Transformers, without sharing their original input data with other devices or servers. Each device trains a model locally in FL and periodically sends it to the server for aggregation.
However, training a complete DNN on end-user devices with relatively limited computing and memory resources, such as smartphones or tablets, presents significant challenges. These devices typically have limited GPU capabilities, making the training of even medium-sized DNNs prohibitively slow. 
Additionally, limited RAM/VRAM on these devices often prevents them from holding both the model parameters and training data in memory. 

Split Federated Learning (SFL)~\cite{splitfed} systems reduce the computational workload of devices in FL by offloading some of the training to a server with more computational resources. In an SFL system, the model is divided into a device block and a server block. The device block processes inputs locally and sends intermediate activations to the server, which computes the loss. The server block then returns the gradients of the activations to the device (Figure~\ref{subfig:sfl}).

\textbf{Large communication overhead:} While SFL systems reduce on-device computation, they introduce communication in addition to the periodic model transfers in FL. Each training iteration exchanges activations and gradients between the device and server, resulting in \textbf{frequent}, \textbf{high-volume} data transfers (red arrows in Figure~\ref{subfig:sfl}). In bandwidth-limited environments, this overhead severely impacts efficiency, making SFL systems impractical for real-world deployments (see Challenge~1 and Challenge~2 discussed in Section~\ref{subsec:challenge}).

\begin{figure}[tp]
    \centering
    \begin{subfigure}{0.48\linewidth}
        \includegraphics[width=0.97\linewidth]{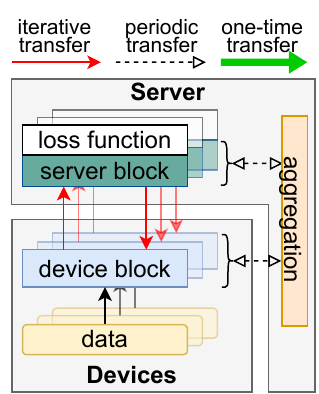}
        \caption{Split Federated Learning}
        \label{subfig:sfl}
    \end{subfigure}
    \begin{subfigure}{0.47\linewidth}
        \includegraphics[width=0.91\linewidth]{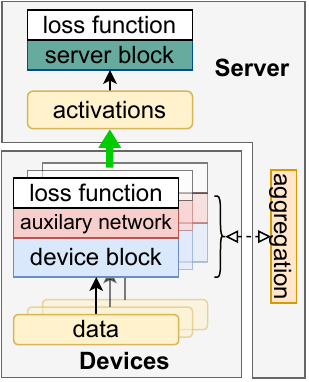}
        \caption{\MethodName}
        \label{subfig:fedgrail}
    \end{subfigure}
    \caption{Training in a Split Federated Learning (SFL) system and \MethodName. SFL exchanges activations and gradients in each iteration; \MethodName\ only transfers activations once.}
    \label{fig:train_diagram}
\end{figure}

Several methods have been proposed to reduce the communication overhead in SFL systems, such as compressing activations and gradients~\cite{bottleneck,c3sl,b2f,topksparse,ecofed}, overlapping communication and computation~\cite{pipar}, and computing local gradients~\cite{splitgp,fedgkt} (discussed further in Section~\ref{sec:rw}). However, these methods still require transferring activations each iteration, preventing them from achieving the communication efficiency of classic FL.

To address the above limitation, we propose \MethodName\ (Figure~\ref{subfig:fedgrail}), a novel system that eliminates iterative activation and gradient transfers in SFL by developing \textbf{unidirectional inter-block training} (see Section~\ref{subsubsec:uit}). Instead of exchanging activations and gradients at each iteration, \MethodName\ first trains the device block to convergence using a \textbf{lightweight auxiliary network} (see Section~\ref{subsubsec:aux}) to compute local loss without needing gradients from the server. Once the device block is trained, activations are transferred to the server only \textbf{once}, allowing the server block to be trained independently. The primary communication overhead of SFL arises from iterative activation and gradient transfers, but \MethodName\ changes this to a single activation transfer (green arrows in Figure~\ref{subfig:fedgrail}; no gradients are transferred from the server). \MethodName\ communication is reduced to the exchanges of device blocks for aggregation, which is a lower communication overhead than classic FL.

\textbf{Degraded accuracy due to data heterogeneity:}
In addition to communication overhead, SFL systems inherit \emph{non-independent and identically distributed (non-IID)} data-related problems from FL. Given heterogeneous data each device will have a different data distribution in FL, leading to model divergence and reduced accuracy. These problems persist in SFL as non-IID data on devices generates skewed activations that negatively impact the training of the server block. While existing FL methods mitigate non-IID problems to some extent using modified aggregation strategies and objective functions~\cite{fedavgm,fedprox,scaffold}, they are less effective in SFL. To address this, \MethodName\ leverages the split architecture to merge device activations into a single, more homogeneous dataset for server block training. This \textbf{activation consolidation} method (see Section~\ref{subsubsec:act}) ensures that the server block trains on more IID inputs, which in turn improves the overall model accuracy.

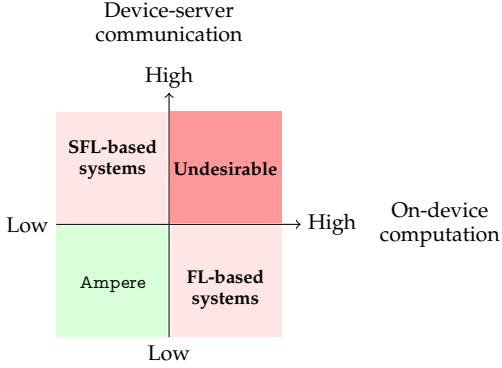
\begin{figure}[t]
    \centering
    \footnotesize
    \begin{tikzpicture}[scale=0.5]
        \draw[->] (-3,0) -- (3.5,0) node at (7.2,0) {\parbox{2cm}{\centering On-device\\computation}};
        \draw[->] (0,-3) -- (0,3.5) 
        node at (0,5.4) {\parbox{2cm}{\centering Device-server communication}};

        \fill[green!15] (-0.05, -0.05) rectangle (-3,-3);
        \node[below] at (-1.5, -1.2) {\scriptsize \MethodName};
        
        \fill[red!10] (0.05, -0.05) rectangle (3,-3);
        \node[below] at (1.5, -1) {\scriptsize \textbf{\parbox{2cm}{\centering FL-based\\systems}}};
        
        \fill[red!10] (-0.05, 0.05) rectangle (-3,3);
        \node[above] at (-1.5, 1) {\scriptsize \textbf{\parbox{2cm}{\centering SFL-based\\systems}}};

        \fill[red!40] (0.03,0.03) rectangle (3,3);
        \node at (1.5, 1.5) {\scriptsize \textbf{Undesirable}};

        \node[left] at (-3,0) {Low};
        \node[right] at (3.5,0) {High};

        \node[below] at (0,-3) {Low};
        \node[below] at (0,4.4) {High};

    \end{tikzpicture}
    \caption{Trade-off between on-device computation and device-server communication for Federated Learning (FL), Split Federated Learning (SFL), and \MethodName. FL and SFL have similar accuracy whereas \MethodName\ has a higher accuracy than these methods.
    }
    \label{fig:position}
\end{figure}

As shown in Figure~\ref{fig:position}, \MethodName\ is the first system to simultaneously \textbf{(i)} reduce communication overhead to FL levels while achieving low on-device computation as in SFL, and \textbf{(ii)} improve model accuracy on non-IID data, surpassing both FL and SFL.
Extensive experimental evaluations conducted on various DNNs and non-IID datasets demonstrate that, compared to state-of-the-art SFL baselines, \MethodName\ incurs up to \rev{911$\times$} lower communication overhead and up to \rev{14.5$\times$} lower on-device computation, achieves a training speedup of up to \rev{18.6$\times$}, and improves model accuracy by up to \rev{11.70 percentage points}\rev{\footnote{\emph{Percentage points} denote absolute, not relative, differences.}}. In addition, \MethodName\ reduces the standard deviation of accuracy for various non-IID degrees by up to \rev{71.13\%} compared to other baselines, demonstrating that it is less impacted by non-IID data. 

The key \textbf{contributions} of this paper are as follows:
\begin{enumerate}[leftmargin=*]
    \item \MethodName, a novel system that is underpinned by unidirectional inter-block training, a new method to significantly reduce both communication and computation overhead while simultaneously improving model accuracy compared to FL and SFL.
    \item A new method to generate a lightweight auxiliary network that breaks the interdependency between the device and server blocks during training, eliminating the iterative transfer of intermediate results. 
    \item An efficient activation consolidation and management strategy to reduce activation heterogeneity that is demonstrated to improve model accuracy.
\end{enumerate}

The rest of this paper is organized as follows. Section~\ref{sec:background} introduces SFL and identifies three opportunities to address critical challenges in SFL. Section~\ref{sec:design} presents \MethodName\ and describes its underlying methods.
Section~\ref{sec:analysis} analyzes model convergence and communication costs of \MethodName. 
Section~\ref{sec:results} evaluates \MethodName\ against relevant baselines. Section~\ref{sec:rw} discusses the related work, and Section~\ref{sec:concl} concludes the paper.

\section{Background}
\label{sec:background}
In this section, we present split federated learning (SFL), identify three key challenges that limit its adoption in the real-world, and present opportunities in collaborative learning beyond SFL.

\subsection{Split Federated Learning}
In a classic FL system, training the complete neural network $\theta$ on devices is prohibitively expensive (computation and memory usage) on end-user devices.

Split Federated Learning (SFL)~\cite{splitfed} systems address these constraints by splitting the model into two blocks -- a device block $\theta^{(d)}$ and a server block $\theta^{(s)}$ that is offloaded to the server. Suppose \(K\) devices participate in training; each device \(k \in [K]\) has its own dataset $\mathcal{D}_k = \left\{(x_i, y_i)\right\}_{i=1}^{n_k}$, with \(x_i\) denoting an input sample, \(y_i\) its label, and \(n_k = |\mathcal{D}_k|\). The overall objective is to train a global model $\theta=(\theta^{(d)},\theta^{(s)})$ that minimizes the following objective function:
\begin{equation}
    \min_\theta F\left(\theta\right) = \sum_{k=1}^{K} \frac{n_k}{n} F_k\left(\theta^{(d)},\theta^{(s)}\right)
    \label{eq:sfl-goal}
\end{equation}
where $n = \sum_{k=1}^{K} n_k$ is the total number of data samples across all devices. For a given device $k$, its local loss function \(F_k\) is defined as:
\begin{equation}
    F_k\left(\theta^{(d)},\theta^{(s)}\right) = \frac{1}{n_k} \sum_{i=1}^{n_k} f^{(s)}\left(\theta^{(s)}_k; f^{(d)}\left(\theta^{(d)}_k; x_i\right), y_i\right)
\end{equation}
where $f^{(d)}\left(\theta^{(d)}_k; x_i\right)$ denotes the device-side forward function, and $f^{(s)}\left(\theta^{(s)}_k; *\right)$ is the server-side loss function that computes the prediction error with respect to the label \(y_i\).

\textbf{Forward and backward passes in SFL:}
During local training on device $k$, the device block $\theta^{(d)}_k$ processes $x_i$ to produce activations $\xi_i=f^{(d)}\left(\theta^{(d)}_k; x_i\right)$, which are sent to the server. The server block $\theta^{(s)}_k$ completes the forward pass and computes the loss. In the backward pass, gradients of $\theta^{(s)}_k$ are obtained and used to update it. Simultaneously, the server sends the corresponding gradients of $\xi$ back to the device, enabling the backward pass of $\theta^{(d)}_k$. $\theta^{(d)}_k$ and $\theta^{(s)}_k$ are updated using stochastic gradient descent (SGD):
\begin{equation}
    \theta^{(d)}_k \leftarrow \theta^{(d)}_k - \eta \nabla F_k\left(\theta^{(d)}_k\right), \enspace 
    \theta^{(s)}_k \leftarrow \theta^{(s)}_k - \eta \nabla F_k\left(\theta^{(s)}_k\right)
\end{equation}
where $\eta$ is the learning rate. \textit{Exchanging activations $\xi_i$ and their gradients in each training iteration adds substantial device-server communication overhead beyond classic FL.}

\textbf{Model aggregation:}
After a user-defined number of local iterations, each device sends its updated $\theta^{(d)}_k$ to the server. The server then aggregates these updates into $\theta^{(d)}$ using the FedAvg algorithm~\cite{fedavg}; meanwhile, the server blocks are aggregated into $\theta^{(s)}$:
\begin{equation}
    \theta^{(d)} \leftarrow \sum_{k=1}^{K} \frac{n_k}{n}\,\theta^{(d)}_k, \enspace 
    \theta^{(s)} \leftarrow \sum_{k=1}^{K} \frac{n_k}{n}\,\theta^{(s)}_k
\end{equation}
Finally, the server broadcasts the updated $\theta^{(d)}$ to all devices for the next round of local training. This process is repeated until convergence. \textit{The non-IID nature of the data locally generated on devices $\{\mathcal{D}_k\}$ reduces model accuracy compared to training on a centralized dataset $\bigcup\limits_{k=1}^K \mathcal{D}_k$}.

\subsection{Opportunities Beyond SFL Systems}
\label{subsec:challenge}

Although SFL systems reduce on-device computation, they face three critical challenges due to additional communication (Challenge~1 and Challenge~2) and data heterogeneity (Challenge~3), which limit their adoption in practice.

\subsubsection{Challenge 1: Trade-Off Between On-Device Computation and Device-Server Communication}
\label{subsubsec:challenge1}

Selecting the optimal layer at which the network must be split, referred to as split point, is crucial as it impacts both on-device computation and device-server communication. 
To illustrate this a MobileNetV3-Large (MobileNet-L)~\cite{mobilenet} model is trained on the CIFAR-10~\cite{cifar10,cifar10-2} dataset (Figure~\ref{fig:computation_communication}). The experimental setup is presented in Section~\ref{subsec:setup}. On-device computation (measured in GFLOPs) and device-server communication volume (shown in GB), including model parameters and intermediate results (activations and gradients), are measured for various split points. Lower split points offload more layers to the server and reduce on-device computation. However, earlier layers generate larger intermediate results, increasing the communication overhead. For example, when compared to classic FL (when all 19 layers are on the device), splitting the network after the first layer \rev{achieves $\sim$56$\times$ lower on-device computation but 9.63$\times$ higher communication}. 
\begin{tcolorbox}[
    width=0.49\textwidth,
    colframe=black!50!black,
    colback=gray!8,boxrule=0.5pt,
    left=2pt,right=2pt,top=2pt,bottom=2pt]
\textbf{Opportunity 1:}
A fundamentally new system is required to eliminate the on-device computation and device-server communication trade-off, as it is impossible to select a split point that minimizes both simultaneously. 
\end{tcolorbox}

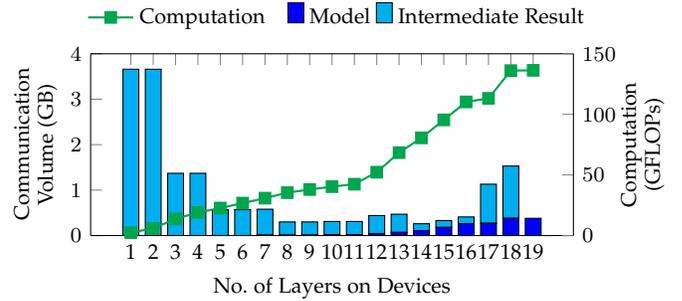
\begin{figure}[tp]
    \centering
   \begin{tikzpicture}[scale=1]
    \begin{axis}[
        xlabel={No. of Layers on Devices},
        xlabel style={align=center, font=\footnotesize},
        height=4cm,
        width=8cm,
        ybar stacked, 
        ymin=0, ymax=4, 
        axis y line*=left,
        xticklabel style={font=\footnotesize},
        yticklabel style={font=\footnotesize},
        ylabel style={at={(-0.05,0.5)}, font=\footnotesize},
        ylabel={\parbox{2cm}{\centering Communication\\Volume (GB)}},
        bar width=6pt,
        enlarge x limits=0.1,
        symbolic x coords={1, 2, 3, 4, 5, 6, 7, 8, 9, 10, 11, 12, 13, 14, 15, 16, 17, 18, 19},
        xtick=data,
        every node near coord/.append style={font=\small},
        legend style={at={(0.39,1.2)}, anchor=west, legend columns=2, font=\footnotesize, draw=none}
    ]
    \addplot[fill=blue] coordinates {
        (1, 0.0)
        (2, 0.0)
        (3, 0.0)
        (4, 0.0)
        (5, 0.0)
        (6, 0.0)
        (7, 0.01)
        (8, 0.01)
        (9, 0.01)
        (10, 0.02)
        (11, 0.02)
        (12, 0.04)
        (13, 0.07)
        (14, 0.11)
        (15, 0.18)
        (16, 0.26)
        (17, 0.27)
        (18, 0.38)
        (19, 0.38)
    };
    \addlegendentry{Model}
    
    \addplot[fill=cyan] coordinates {
        (1, 3.66)
        (2, 3.66)
        (3, 1.37)
        (4, 1.37)
        (5, 0.57)
        (6, 0.57)
        (7, 0.57)
        (8, 0.29)
        (9, 0.29)
        (10, 0.29)
        (11, 0.29)
        (12, 0.4)
        (13, 0.4)
        (14, 0.15)
        (15, 0.15)
        (16, 0.15)
        (17, 0.86)
        (18, 1.15)
        (19, 0)
    };
    \addlegendentry{Intermediate Result}

    \end{axis}

    \begin{axis}[
        height=4cm,
        width=8cm,
        axis y line*=right,
        axis x line=none,
        enlarge x limits=0.1,
        symbolic x coords={1, 2, 3, 4, 5, 6, 7, 8, 9, 10, 11, 12, 13, 14, 15, 16, 17, 18, 19},
        ymin=0, ymax=150,
        xticklabel style={font=\footnotesize},
        yticklabel style={font=\footnotesize},
        ylabel={Model Accuracy (\%)},
        y label style={at={(axis description cs:1.08,0.5)},font=\footnotesize},
        ylabel={\parbox{2cm}{\centering On-Device Computation\\(GFLOPs)}},
        legend style={at={(0.39,1.2)}, anchor=east, legend columns=2, font=\footnotesize, draw=none},
    ]
    \addplot[mark=square*, Green, thick] coordinates {
        (1, 2.44)
        (2, 5.94)
        (3, 13.71)
        (4, 19.02)
        (5, 22.65)
        (6, 26.75)
        (7, 30.84)
        (8, 35.38)
        (9, 37.97)
        (10, 40.35)
        (11, 42.37)
        (12, 52.26)
        (13, 68.41)
        (14, 80.59)
        (15, 95.43)
        (16, 110.27)
        (17, 113.21)
        (18, 136.08)
        (19, 136.32)
    };
    \addlegendentry{Computation}
    \end{axis}
    \end{tikzpicture}
    \caption{Trade-off between communication volume and on-device computation per round of MobileNet-Large using an SFL system for different split points.}
    \label{fig:computation_communication}
\end{figure}

\subsubsection{Challenge 2: High Device-Server Communication Volume and Frequency}
\label{subsubsec:challenge2}

The split point in SFL is usually selected at the early layers of the model to ensure that the device block remains lightweight and can be trained efficiently on resource-constrained devices~\cite{fedadapt,pipar}. However, this leads to a significant increase in the \textit{communication volume} compared to FL, as large intermediate activations must be transferred from devices to the server during each iteration.

Additionally, end-to-end training of the device and server blocks is tightly interdependent, as the device block relies on the global loss function computed on the server. This leads to a high \textit{communication frequency} that cannot be reduced, making SFL impractical in bandwidth-constrained environments.

We compare the communication volume and frequency for training a MobileNet-L model on CIFAR-10 using FL and SFL with a split point at the first layer. Both methods train for 150 epochs until convergence. Each time a model or a batch of activations/gradients is transferred, it is counted as a communication round. As shown in Table~\ref{tab:challenge-comm}, SFL increases the total communication volume by one order of magnitude and the communication frequency by three orders of magnitude compared to FL.

\begin{table}
    \centering
    \caption{Communication volume and frequency in FL and SFL for MobileNet-L.}
    \label{tab:challenge-comm}
    \begin{tabular}{c|c|c}
        \hline
        \colorrow
         \textbf{System} & \makecell{\textbf{Communication}\\ \textbf{volume (GB)}} & \makecell{\textbf{Communication}\\ \textbf{rounds per hour}} \\
         \hline
         FL & 57 & 768  \\
         SFL & 549 & 3578280  \\
         \hline
    \end{tabular}
\end{table}

\begin{tcolorbox}[
    width=0.49\textwidth,
    colframe=black!50!black,
    colback=gray!8,boxrule=0.5pt,
    left=2pt,right=2pt,top=2pt,bottom=2pt]
\textbf{Opportunity 2:}
Alternative methods are required to break the iterative dependency between the server and device arising from end-to-end training, without relying on a global loss function.
\end{tcolorbox}

\subsubsection{Challenge 3: Data Heterogeneity Reduces Model Accuracy}
\label{subsubsec:challenge3}

In classic FL, each device’s dataset is typically non-independent and non-identically distributed (non-IID), causing local models to learn device-specific patterns that do not generalize well when aggregated. This reduces the global model accuracy compared to training on IID data.

In SFL, each device processes non-IID data, generating skewed activations. The server typically trains multiple server blocks, each on activations from a single device. Similar to device blocks, the server blocks are aggregated periodically. As a result, server blocks learn from skewed activations that represent feature abstractions of the non-IID input data, reducing model accuracy. As data heterogeneity increases and more devices participate, the activation is skewed further, which degrades accuracy.

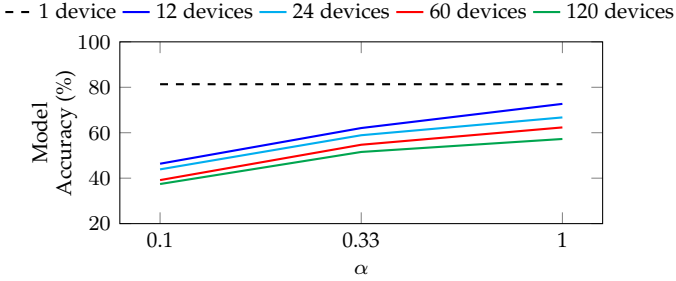
\begin{figure}
    \centering
    \begin{tikzpicture}
    \begin{axis}[
        width=8cm,
        height=4cm,
        xlabel={$\alpha$},
        ylabel={\parbox{2cm}{\centering Model\\Accuracy (\%)}},
        symbolic x coords={0.1,0.33,1},
        xtick=data,
        ymin=20, ymax=100,
        legend style={at={(0.46,1.05)}, anchor=south,legend columns=5, font=\footnotesize, draw=none},
        legend image post style={xscale=0.6},
        xticklabel style={font=\footnotesize},
        yticklabel style={font=\footnotesize},
        xlabel style={align=center, font=\footnotesize},
        ylabel style={at={(-0.07,0.5)}, font=\footnotesize},
    ]
    
    \addplot[
        color=black,
        thick,
        dashed,
    ]
    table[
        x=alpha,
        y=1_device,
        col sep=comma
    ]{acc_vs_noniid.csv};
    \addlegendentry{1 device}

    \addplot[
        color=blue,
        solid,
        thick,
    ]
    table[
        x=alpha,
        y=12_devices,
        col sep=comma
    ]{acc_vs_noniid.csv};
    \addlegendentry{12 devices}

    \addplot[
        color=cyan,
        solid,
        thick,
    ]
    table[
        x=alpha,
        y=24_devices,
        col sep=comma
    ]{acc_vs_noniid.csv};
    \addlegendentry{24 devices}

    \addplot[
        color=red,
        solid,
        thick,
    ]
    table[
        x=alpha,
        y=60_devices,
        col sep=comma
    ]{acc_vs_noniid.csv};
    \addlegendentry{60 devices}

    \addplot[
        color=Green,
        solid,
        thick,
    ]
    table[
        x=alpha,
        y=120_devices,
        col sep=comma
    ]{acc_vs_noniid.csv};
    \addlegendentry{120 devices}
    
    \end{axis}
    \end{tikzpicture}
    \caption{Model accuracy of MobileNet-Large for varying number of devices and non-IID degree $\alpha$.}
    \label{fig:acc_noniid}
\end{figure}

Data heterogeneity is modeled using the CIFAR-10 dataset, where the data is partitioned across multiple devices following a Dirichlet distribution $Dir(\frac{\alpha}{1-\alpha+\epsilon})$ (see Section~\ref{subsec:setup}). A smaller $\alpha$ results in a higher non-IID degree, whereas $\alpha=1$ approximates an IID scenario. Figure~\ref{fig:acc_noniid} shows that model accuracy decreases when $\alpha$ gets smaller because the data on individual devices becomes more heterogeneous. 
Additionally, for a given $\alpha$, increasing the number of devices further increases heterogeneity and reduces accuracy. The dashed line in the figure shows the accuracy achieved when training on the entire dataset (which is IID), serving as an upper bound of model accuracy that can be achieved under any non-IID configuration.
\begin{tcolorbox}[
    width=0.49\textwidth,
    colframe=black!50!black,
    colback=gray!8,boxrule=0.5pt,
    left=2pt,right=2pt,top=2pt,bottom=2pt]
\textbf{Opportunity 3:}
A new system is required to enhance server block training by reducing activation heterogeneity from device blocks to improve overall model accuracy. 
\end{tcolorbox}

\section{\texorpdfstring{\texttt{A\MakeLowercase{mpere}}}{Ampere}}
\label{sec:design}
In this section, we leverage the three opportunities identified in the previous section to design a new system \MethodName. Three new methods, namely unidirectional inter-block training, lightweight auxiliary network generation, and efficient activation consolidation are developed and incorporated in \MethodName. The training algorithm incorporating these methods is also considered. 

\subsection{Overview}

\begin{figure*}[tp]
    \centering
    \includegraphics[width=\linewidth]{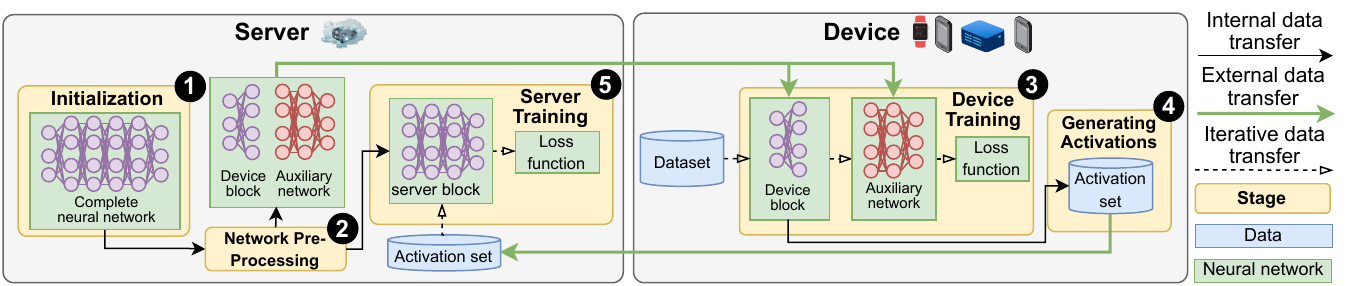}
    \caption{Overview of \MethodName.}
    \label{fig:overview}
\end{figure*}

\MethodName\ operates with a single server and $K$ devices, collaboratively training a deep neural network $\theta$ via \textbf{unidirectional inter-block training (UIT)} instead of training based on conventional end-to-end backpropagation. Specifically, \MethodName\ splits $\theta$ into a device block and a server block, first training the device block until convergence before proceeding to train the server block. This process follows five key steps (see Figure~\ref{fig:overview}):

Step {\color{white}\cir{1}}: Initialization - the server initializes the complete neural network $\theta$.

Step {\color{white}\cir{2}}: Network Pre-processing - $\theta$ is split into two blocks:
(i)~Device block $\theta^{(d)}$, comprising the initial layers of $\theta$.
(ii)~Server block $\theta^{(s)}$, containing the remaining layers.
An auxiliary network $\tilde{\theta}^{(d)}$ is generated to connect the outputs of $\theta^{(d)}$ to a local loss function. The server then transfers $\theta^{(d)}$ and $\tilde{\theta}^{(d)}$ to all participating devices, denoted as $\theta^{(d)}_k$ and $\tilde{\theta}^{(d)}_k$ for device $k$.

Step {\color{white}\cir{3}}: Device Training - devices train $\left\{\theta^{(d)}_k, \tilde{\theta}^{(d)}_k\right\}$ locally until convergence using the auxiliary network and local loss function. Unlike SFL, gradients of device blocks are computed locally in \MethodName\ without sending intermediate activations to the server.

Step {\color{white}\cir{4}}: Generating Activations - each device generates activations by feeding local data into its trained device block $\theta^{(d)}$. These activations are \textbf{transferred once} to the server, where they are consolidated for training the server block.  

Step {\color{white}\cir{5}}: Server Training - $\theta^{(s)}$ is trained on the server using the consolidated activations, without interacting with devices.  

The above sequential steps execute only once, eliminating the iterative back-and-forth device-server communication in SFL.

\subsection{Methods}
The three methods developed in \MethodName\ are presented.

\subsubsection{Unidirectional Inter-Block Training to Eliminate the On-Device Computation and Device-Server Communication Trade-Off}
\label{subsubsec:uit}

Classic SFL relies on end-to-end backpropagation (BP)-based training. In each training iteration, data is transferred from the device block to the server block during the forward pass, while gradients propagate back during the backward pass. This results in a fundamental trade-off between on-device computation and device-server communication, which cannot be simultaneously optimized at the same split point as presented in Section~\ref{subsubsec:challenge1}.

In contrast, \MethodName\ develops \textbf{unidirectional inter-block training (UIT)}, treating the device block and the server block as independent modules trained sequentially rather than iteratively. UIT transfers activations from the device to the server only once during training and no gradients are sent from the server to the device. This approach eliminates the aforementioned trade-off by ensuring that on-device computation and device-server communication are optimized simultaneously at the same split point.

Split point $p$ refers to the number of layers in the device block. Similar to BP, on-device computation in UIT increases with the number of layers in the device block $\theta^{(d)}$ (increasing $p$). Thus, placing only the first layer on the device ($p=1$) reduces computation.

However, the trend of UIT communication with respect to the split point in \MethodName\ differs from BP. UIT communication consists of: (i)~Model exchanges for aggregating device blocks in Step~{\color{white}\cir{3}}. (ii)~A one-time activation transfer in Step~{\color{white}\cir{4}}.

We measure the parameter size $s^l_i$ and output (activation) size $s^o_i$ for each layer $i \in [1, I]$, where $I$ is the total number of layers. The total communication volume $C$ over $N$ training epochs is:
\begin{equation}
    C = 2N \sum_{i=1}^p s^l_i + s^o_p
    \label{eq:comm}
\end{equation}
where the first term accounts for model exchanges occurring twice per epoch, while the second term represents the one-time activation transfer. Since training typically requires $N \ge 100$ epochs, the activation transfer cost $s^o_p$ becomes negligible (see Section~\ref{subsec:comm_analysis} for a detailed analysis), and $C$ is dictated by $\sum_{i=1}^p s^l_i$, which increases with $p$. Thus, $p=1$ minimizes the device-server communication.

Figure~\ref{fig:computation_communication_fedgrail} shows the communication volume and on-device computation per training round for MobileNet-Large on CIFAR-10. For BP the lowest amount of on-device computation and device-server communication incurred are achieved at different split points, $p=1$ and $p=14$, respectively. In contrast, UIT demonstrates a consistent trend -- both increase with $p$. As a result, the optimal split point $p_{opt}=1$ of UIT simultaneously minimizes both computation and communication, thereby addressing \textbf{Challenge 1}. \rev{An empirical evaluation comparing the final model accuracy, training time, and communication volume for different split points is presented in Section~\ref{subsubsec:split_ablation}, confirming that $p=1$ selected in this work achieves competitive accuracy while minimizing costs.}

\begin{figure}[tp]
    \centering
   \begin{tikzpicture}[scale=1]

    \begin{axis}[
        xlabel={Split Point},
        xlabel style={align=center, font=\footnotesize},
        height=4cm,
        width=8cm,
        axis y line*=left,
        enlarge x limits=0.1,
        symbolic x coords={1, 2, 3, 4, 5, 6, 7, 8, 9, 10, 11, 12, 13, 14, 15, 16, 17, 18, 19},
        ymin=0, ymax=4,
        xticklabel style={font=\footnotesize},
        yticklabel style={font=\footnotesize},
        ylabel={Model Accuracy (\%)},
        ylabel style={at={(-0.05,0.5)}, font=\footnotesize},
        ylabel={\parbox{2cm}{\centering Communication\\Volume (GB)}},
        xtick=data,
        every node near coord/.append style={font=\small},
        legend style={at={(0.23,1.2)}, anchor=west, legend columns=2, font=\scriptsize, draw=none}
    ]
    \addplot[Green, thick] coordinates {
        (1, 0.0366)
        (2, 0.0366)
        (3, 0.0137)
        (4, 0.0137)
        (5, 0.0057)
        (6, 0.0057)
        (7, 0.0129)
        (8, 0.0129)
        (9, 0.0129)
        (10, 0.0229)
        (11, 0.0229)
        (12, 0.044)
        (13, 0.074)
        (14, 0.1115)
        (15, 0.1815)
        (16, 0.2615)
        (17, 0.2786)
        (18, 0.3915)
        (19, 0.38)
    };
    \addlegendentry{UIT Communication}

    \addplot[Red, thick] coordinates {
        (1, 3.66)
        (2, 3.66)
        (3, 1.37)
        (4, 1.37)
        (5, 0.57)
        (6, 0.57)
        (7, 0.58)
        (8, 0.3)
        (9, 0.3)
        (10, 0.31)
        (11, 0.31)
        (12, 0.44)
        (13, 0.47)
        (14, 0.26)
        (15, 0.33)
        (16, 0.41)
        (17, 1.13)
        (18, 1.53)
        (19, 0.38)
    };
    \addlegendentry{BP Communication}
    \end{axis}

    \begin{axis}[
        height=4cm,
        width=8cm,
        axis y line*=right,
        axis x line=none,
        enlarge x limits=0.1,
        symbolic x coords={1, 2, 3, 4, 5, 6, 7, 8, 9, 10, 11, 12, 13, 14, 15, 16, 17, 18, 19},
        ymin=0, ymax=150,
        xticklabel style={font=\footnotesize},
        yticklabel style={font=\footnotesize},
        ylabel={Model Accuracy (\%)},
        y label style={at={(axis description cs:1.08,0.5)},font=\footnotesize},
        ylabel={\parbox{3cm}{\centering On-Device Comp.\\(GFLOPs)}},
        legend style={at={(0.255,1.2)}, anchor=east, legend columns=2, font=\scriptsize, draw=none}
    ]
    \addplot[Blue, thick] coordinates {
        (1, 2.44)
        (2, 5.94)
        (3, 13.71)
        (4, 19.02)
        (5, 22.65)
        (6, 26.75)
        (7, 30.84)
        (8, 35.38)
        (9, 37.97)
        (10, 40.35)
        (11, 42.37)
        (12, 52.26)
        (13, 68.41)
        (14, 80.59)
        (15, 95.43)
        (16, 110.27)
        (17, 113.21)
        (18, 136.08)
        (19, 136.32)
    };
    \addlegendentry{Computation}
    \end{axis}
    \end{tikzpicture} 
    \caption{Communication volume and on-device computation per training round for MobileNet-Large using BP and UIT for different split points.}
    \label{fig:computation_communication_fedgrail}
\end{figure}
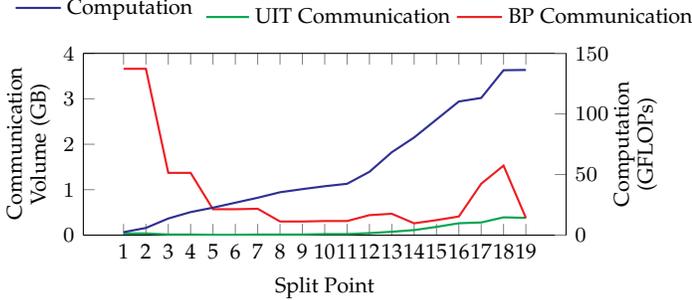

\subsubsection{Lightweight Auxiliary Network Generation to Break Iterative Dependency Between Device and Server Training}
\label{subsubsec:aux}

When using BP in SFL, the device and server blocks are optimized using a global loss function. This is the primary reason for the tight interdependency between the device and server due to frequent iterative transfers of activations and gradients. A large communication volume and frequency is the resultant as highlighted in Section~\ref{subsubsec:challenge2}.

To break this dependency, \MethodName\ introduces a method to \textbf{generate an auxiliary network} $\tilde{\theta}^{(d)}$ that enables independent training of $\theta^{(d)}$. This network connects $\theta^{(d)}$ to a local loss function, allowing devices to compute gradients without requiring gradients from the server.

Since $\tilde{\theta}^{(d)}$ runs on resource-constrained devices, it must be lightweight to keep the computational overhead low while capturing meaningful feature representations for training the server block. To achieve this, a compact two-layer design is followed. The first layer replicates the first layer in $\theta^{(s)}$ but halves the dimension. The second layer is a fully connected layer with the same loss function as in $\theta^{(s)}$.

The second layer of $\tilde{\theta}^{(d)}$ provides a direct mapping to the loss function so that local training is feasible. The first layer ensures $\theta^{(d)}$ captures generalizable features from raw data, which is essential for the effective learning of the subsequent layers in the server block~\cite{zeiler2014visualizing}. Without this layer, $\tilde{\theta}^{(d)}$ would rely on a fully connected layer, forcing $\theta^{(d)}$ to extract task-specific abstract representations. However, since the server block expects the device block to supply shallow, general-purpose features, such task-specific representations do not generalize well across devices and degrade overall model performance~\cite{yosinski2014transferable,lecun2015deep}.

The dimension of both layers in $\tilde{\theta}^{(d)}$ impacts both model accuracy and computational cost. The method that generates the auxiliary network aims to reduce on-device computation, for which the auxiliary network’s dimension is scaled down relative to the first layer of the server block. Figure~\ref{fig:aux_dimension} demonstrates the effect of different dimension ratios of the auxiliary network relative to the first layer of the server block on MobileNet-Large trained on CIFAR-10. As the dimension ratio increases from 0.25 to 1.0, computational overhead grows linearly, while accuracy improvements become marginal beyond 0.5. Therefore, 0.5 is the selected ratio in \MethodName\ but allows users to adjust it.

\begin{figure}[tp]
    \centering
   \begin{tikzpicture}[scale=1]

    \begin{axis}[
        xlabel={Dimension Ratio of Auxiliary Network},
        xlabel style={align=center, font=\footnotesize},
        height=4cm,
        width=8cm,
        axis y line*=left,
        enlarge x limits=0.1,
        symbolic x coords={0.25,0.5,0.75,1.0},
        ymin=50, ymax=70,
        xticklabel style={font=\footnotesize},
        yticklabel style={font=\footnotesize},
        ylabel style={at={(-0.05,0.5)}, font=\footnotesize},
        ylabel={\parbox{2cm}{\centering Model\\Accuracy (\%)}},
        xtick=data,
        every node near coord/.append style={font=\small},
        legend style={at={(0.5,1.2)}, anchor=west, legend columns=2, font=\footnotesize, draw=none}
    ]
    \addplot[Green, thick] coordinates {
        (0.25, 62.42)
        (0.5, 67.3375)
        (0.75, 67.6375)
        (1.0, 68.35)
    };
    \addlegendentry{Model Accuracy}
    \end{axis}

    \begin{axis}[
        height=4cm,
        width=8cm,
        axis y line*=right,
        axis x line=none,
        enlarge x limits=0.1,
        symbolic x coords={0.25,0.5,0.75,1.0},
        ymin=0, ymax=7,
        xticklabel style={font=\footnotesize},
        yticklabel style={font=\footnotesize},
        y label style={at={(axis description cs:1.05,0.5)},font=\footnotesize},
        ylabel={\parbox{2cm}{\centering Computation\\(GFLOPs)}},
        legend style={at={(0.5,1.2)}, anchor=east, legend columns=2, font=\footnotesize, draw=none}
    ]
    \addplot[Red, thick] coordinates {
        (0.25, 1.44)
        (0.5, 2.89)
        (0.75, 4.33)
        (1.0, 5.78)
    };
    \draw[dashed] (axis cs:0.5,0) -- (axis cs:0.5,100);
    \addlegendentry{Computation}
    \end{axis}
    \end{tikzpicture} 
    \caption{Impact of varying auxiliary network dimensions on on-device computation per round and final model accuracy. The x-axis represents the dimension ratio of auxiliary network relative to the first layer of the server block.}
    \label{fig:aux_dimension}
\end{figure}
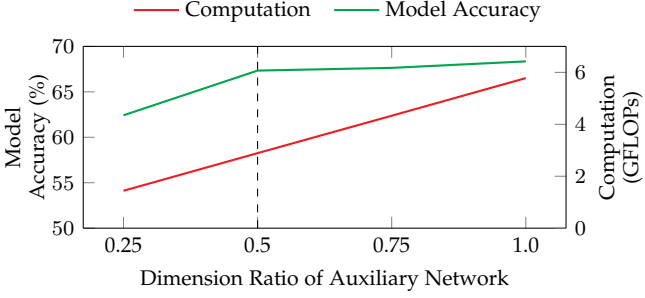

\rev{A study varying the number of layers in the auxiliary network (Section~\ref{subsubsec:aux_ablation}) confirms that a single fully connected layer design underfits, while a deeper auxiliary network improves accuracy only at a disproportionate computational and communication cost. The 2-layer design has the best trade-off across CNN and Transformer architectures.}

By decoupling device and server training using a lightweight auxiliary network, \MethodName\ reduces iterative activation transfers to a one-shot transfer, reducing communication overhead and addressing \textbf{Challenge~2}.

\subsubsection{Activation Consolidation to Mitigate Impact of Data Heterogeneity}
\label{subsubsec:act}

\MethodName\ for the first time addresses activation heterogeneity when collaboratively training multiple devices by introducing \textbf{activation consolidation} on the server. Unlike SFL systems, where the server block is trained on the activations generated from non-IID local data across devices, \MethodName\ trains the server block on the activations produced by a converged and consistent device block using data from all devices. This mitigates the impact of heterogeneous activations and improves model accuracy. 

In a typical SFL system, each device $k$ generates activations from its local dataset $\mathcal{D}_k$, leading to skewed intermediate feature representations across different devices. When the server trains $K$ separate server blocks, each on the activations from a single device, the non-IID activation distributions lead to inconsistent feature learning across server blocks, reducing generalization and overall model accuracy. SplitFedV2, a variant of SplitFed~\cite{splitfed}, removes server-side model aggregation by training a single server block. However, the server block is still trained on activations uploaded in every training iteration, although the device blocks are often not fully trained. As the parameters of the device blocks continue to change during each training iteration, the resulting activations for the same input will also change. This adversely impacts the server block's accuracy. 

To mitigate the impact of activation heterogeneity, \MethodName\ takes a fundamentally different approach by performing one-shot activation consolidation after the convergence of the device block. Each device generates activations using the same converged device block and uploads them once to the server. The server then forms a unified activation set $\mathcal{A}$:
\begin{equation}
    \mathcal{A}=\left\{\left(\xi_i,y_i\right);\xi_i=f^{d}\left(\theta^{(d)}; x_i\in \bigcup\limits_{k=1}^{K} \mathcal{D}_k\right)\right\}
    \label{eq:act_compute}
\end{equation}
where $\xi_i$ represents the activation generated from input $x_i$ by the converged device block $f^{d}(\theta^{(d)})$, and $y_i$ is its label. Subsequently, a single server block is trained on the unified activation set $\mathcal{A}$. 

This approach reduces activation heterogeneity and improves server-side generalization over existing SFL systems for two reasons. First, since all activations are generated using the same converged device block, the representation mapping is identical across devices. This eliminates inconsistencies that would otherwise arise from model parameter variations. Second, by consolidating activations from all devices, the unified activation set $\mathcal{A}$ better approximates the global data distribution. This improves the diversity and balance of the training set seen by the server block, enabling it to learn more generalizable representations.

\subsection{Training Algorithm}

Training in \MethodName\ follows a dual-step approach (as shown in Algorithm~\ref{alg:train}): (1) the device block is trained, similar to FL, alongside an additional auxiliary network, and (2) the server trains a single server block in a centralized manner. 

\textbf{Device Training}: The device block $\theta^{(d)}$ is trained locally with the auxiliary network $\tilde{\theta}^{(d)}$. On each device $k$, $\theta^{(d)}$ and $\tilde{\theta}^{(d)}$ are trained iteratively, minimizing the objective:
\begin{equation}
    \min_{\theta^{(d)}} F^{(d)}(\theta^{(d)}) = \sum_{k=1}^{K} \frac{n_k}{n} F^{(d)}_k(\theta^{(d)},\tilde{\theta}^{(d)})
\end{equation}
where $F^{(d)}_k$ is the local loss function on device $k$: 
\begin{equation}
F^{(d)}_k\left(\theta^{(d)},\tilde{\theta}^{(d)}\right) = \frac{1}{n_k} \sum_{i=1}^{n_k} \tilde{f}^{(d)}\left(\tilde{\theta}^{(d)}_k;f^{(d)}\left(\theta^{(d)}_k; x_i\right), y_i\right)
\end{equation}

Each iteration involves sampling a mini-batch from the local dataset (Line~4) and computing activations $\xi_i$ via the feed-forward function $f^{(d)}$. The auxiliary network $\tilde{f}^{(d)}$ then computes the local loss and updates $\theta^{(d)}$ and $\tilde{\theta}^{(d)}$ (Line~5): 
\begin{equation}
\begin{split}
    \theta^{(d)}_k \leftarrow \theta^{(d)}_k - \eta^{(d)} \nabla F^{(d)}_k\left(\theta^{(d)}_k\right)\\
    \tilde{\theta}^{(d)}_k \leftarrow \tilde{\theta}^{(d)}_k - \eta^{(d)} \nabla F^{(d)}_k\left(\tilde{\theta}^{(d)}_k\right)
    \label{eq:fedgrail_update}
\end{split}
\end{equation}
where $\eta^{(d)}$ is the learning rate of the device block.

The updated local models are then sent to the server (Line~6) and aggregated using (Lines~13-14): 
\begin{equation}
    \theta^{(d)} \leftarrow \sum_{k=1}^{K} \frac{n_k}{n} \theta^{(d)}_k, \enspace 
    \tilde{\theta}^{(d)} \leftarrow \sum_{k=1}^{K} \frac{n_k}{n} \tilde{\theta}^{(d)}_k
    \label{eq:fedgrail_agg}
\end{equation}

The global model is then sent to the devices (Line~15) and the local models (Line~8) are updated. This repeats until $\left\{\theta^{(d)},\tilde{\theta}^{(d)}\right\}$ converges or reaches $N^{(d)}$ epochs.

\textbf{Activation Transfer}: 
Once device block training completes, each device generates activations by feeding local data into $\theta^{(d)}$ and sends them with their labels to the server (Line~10), consolidating an activation set $\mathcal{A}$ (Equation~\labelcref{eq:act_compute}).

\textbf{Server Training}: 
The server performs two tasks: (i) storing activations from the devices on disk (Line~16) and (ii) loading activation batches for training the server block (Lines~17-19). Since devices send activations at different rates, waiting for all activations before training the server block introduces idle time. To mitigate this, these tasks run asynchronously in parallel, allowing training to begin as soon as the first batch of activations is available, improving training efficiency.

The server block $\theta^{(s)}$ is trained on the consolidated activation set to optimize the following objective:
\begin{equation}
    \min_{\theta^{(s)}} F^{(s)}\left(\theta^{(s)}\right) = \frac{1}{|\mathcal{A}|} \sum_{i=1}^{|\mathcal{A}|} f^{(s)}\left(\theta^{(s)}; \xi_i, y_i\right)
    \label{eq:goal}
\end{equation}
where $f^{(s)}$ is the loss function on server. $\theta^{(s)}$ is updated via: 
\begin{equation}
    \theta^{(s)} \leftarrow \theta^{(s)} - \eta^{(s)} \nabla F^{(s)}\left(\theta^{(s)}\right)
    \label{eq:server_update}
\end{equation}
where $\eta^{(s)}$ is the learning rate of the server block. Training continues until convergence or reaches $N^{(s)}$ epochs.


\begin{algorithm}[t]
    \caption{Training in \MethodName}
	\label{alg:train}
    
    \SetKwFunction{FMain}{DeviceTraining}
    \SetKwProg{Fn}{Function}{\string:}{}
    
    \tcc{Training device blocks in parallel}
    \Fn{\FMain{$k$, $\mathcal{D}_k$, $N^{(d)}$, $H$, $B^{(d)}$}}{
        \KwIn{device index $k\in[K]$, dataset $\mathcal{D}_k$, \#device epochs $N^{(d)}$, \#iterations per epoch $H$, batch size $B^{(d)}$}
        \KwOut{trained device block $\theta^{(d)}$}
        
        \For{$n \gets 0$ \KwTo $N^{(d)}$}{
            \For{$h \gets 0$ \KwTo $H$}{
                Sample mini-batch $\left\{\bm{x}_h,\bm{y}_h\right\} \in \mathcal{D}_k$ of size $B^{(d)}$
                
                Update $\theta^{(d)}_k$ and $\tilde{\theta}^{(d)}_k$ (Equation~\labelcref{eq:fedgrail_update})
            }
            Send $\left\{\theta^{(d)}_k,\tilde{\theta}^{(d)}_k\right\}$ to server
            
            Receive global models $\left\{\theta^{(d)},\tilde{\theta}^{(d)}\right\}$
            
            Update local models: $\left\{\theta^{(d)}_k,\tilde{\theta}^{(d)}_k\right\} \leftarrow \left\{\theta^{(d)},\tilde{\theta}^{(d)}\right\}$
        }
        
        \tcp{Send activations for server training}
        \ForEach{$\left\{\bm{x}_i,\bm{y}_i\right\} \in \mathcal{D}_k$}{
            Compute activations $\bm{\xi}_i$ (Equation~\labelcref{eq:act_compute}) and send to server
        }
    }
    
    \BlankLine
    \tcc{Training the server block}
    \SetKwFunction{FMain}{ServerTraining}
    \SetKwProg{Fn}{Function}{\string:}{}
    \Fn{\FMain{$N^{(d)},N^{(s)},B^{(s)}$}}{
        \KwIn{\#device epochs $N^{(d)}$, \#server epochs $N^{(s)}$, batch size $B^{(s)}$}
        \KwOut{trained server block $\theta^{(s)}$}
        
        \tcp{Aggregate device blocks}
        \For{$n \gets 0$ \KwTo $N^{(d)}$}{
            Receive $\left\{\theta^{(d)}_k,\tilde{\theta}^{(d)}_k;k\in[K]\right\}$ from devices
            
            Aggregate into $\theta^{(d)}$ and $\tilde{\theta}^{(d)}$ (Equation~\labelcref{eq:fedgrail_agg})
            
            Send $\left\{\theta^{(d)},\tilde{\theta}^{(d)}\right\}$ to devices
        }
        
        \tcp{Subprocess 1: Receive and store activations}        
        Run \texttt{StoreActivation}($\left\{\bm{\xi}_i,\bm{y}_i\right\}$)
        
        \tcp{Subprocess 2: Load activations to train the server block}
        \For{$n \gets 0$ \KwTo $N^{(s)}$}{
            \ForEach{$\left\{\bm{\xi}_i,\bm{y}_i\right\} \in$ \textnormal{\texttt{LoadActivation}()}}{
                Update $\theta^{(s)}$ (Equation~\labelcref{eq:server_update})
            }
        }
    }
\end{algorithm}

\section{Convergence and Communication Cost Analysis}
\label{sec:analysis}
This section analyzes the convergence of \MethodName, and demonstrates that \MethodName\ incurs lower communication costs than both FL and SFL.

\subsection{Convergence Analysis}
\label{subsec:converge}

The convergence of the device and server blocks is analyzed based on the general assumptions of FL and SFL.

\subsubsection{Notation} 
The parameters of device blocks and the auxiliary networks are denoted as 
$\varphi=\left(\theta^{(d)},\,\tilde{\theta}^{(d)}\right)$. $\varphi_k^t$ denotes the parameters of $t$ iterations on device $k$, where $t\in\left[T^{(d)}\right],T^{(d)}=N^{(d)}H$. The objective of device block training becomes:
\begin{equation}
    \min F^{(d)}(\varphi) = \sum_{k=1}^{K} \frac{n_k}{n} F^{(d)}_k(\varphi)
\end{equation}

$F^{(d)}_k$ is the local loss function on device $k$, rewritten as: 
\begin{equation}
F^{(d)}_k\left(\varphi\right) = \frac{1}{n_k} \sum_{i=1}^{n_k} f(\varphi_k;x_i,y_i)
\end{equation}
where
\begin{equation}
f(\varphi_k;x_i,y_i) =\tilde{f}^{(d)}\left(\tilde{\theta}^{(d)}_k;f^{(d)}\left(\theta^{(d)}_k; x_i\right), y_i\right)
\end{equation}

We assume the dataset $\mathcal{D}_k$ on device $k\in[K]$ follows the distribution $p_k(x)$. The non-IID degree is quantified as:
\begin{equation}
    \Gamma=F^{(d),*}-\sum_{k=1}^K \frac{n_k}{n} F^{(d),*}_k
\end{equation}
where $F^{(d),*}$ and $F^{(d),*}_k$ denotes the minimum values of $F^{(d)}$ and $F^{(d)}_k$. Larger value of $\Gamma$ indicates a higher degree of non-IID; $\Gamma=0$ indicates IID data.

The distribution of $\xi$ is denoted as $q(\xi)$, where $\xi=f^{d}(x)$. Let $q^{*}(\xi)$ denote the density of the activations from the converged device block. We measure the distance between $q(\xi)$ and $q^{*}(\xi)$ as:
\begin{equation}
\label{eq:divergence}
    c=\int|q(\xi)-q^*(\xi)|d\xi
\end{equation}

For the sake of simplicity, we use $\theta_t$ to denote the parameters of the server block at iteration $t\in\left[T^{(s)}\right]$, replacing $\theta^{(s)}$, where $T^{(s)}=N^{(s)}\lfloor \frac{n}{B^{(s)}}\rfloor$ denotes the total number of iterations of training the server block.

\subsubsection{Assumptions} 

The following general assumptions are made in the literature on FL and SFL systems~\cite{Li2020On,decoupled,doi:10.1137/16M1080173,NEURIPS2018_a36b598a,10.5555/2567709.2567769,stich2018local,10.5555/3327345.3327357,10.1609/aaai.v33i01.33015693}:

\begin{assumption}[$L$-smoothness~\cite{Li2020On,doi:10.1137/16M1080173,NEURIPS2018_a36b598a}]
\label{assump:lsmooth}
For $\forall k\in[K]$, $F^{(d)}_k(\varphi)$ and $F^{(s)}(\theta)$ are $L$-smooth, i.e., $\exists L$, $\forall \varphi,\varphi^\prime$, such that
\begin{equation}
    \bigl\|\nabla F^{(d)}_k(\varphi) - \nabla F^{(d)}_k(\varphi^\prime)\bigr\| \le L\bigl\|\varphi - \varphi^\prime\bigr\|
\end{equation}
and $\forall \theta,\theta^\prime$, such that
\begin{equation}
    \bigl\|\nabla F^{(s)}(\theta) - \nabla F^{(s)}(\theta^\prime)\bigr\| \le L\bigl\|\theta - \theta^\prime\bigr\|
\end{equation}
Consequently, $F^{(d)}(\varphi)$ is also $L$-smooth.
\end{assumption}

\begin{assumption}[$\mu$-strong convexity~\cite{Li2020On}]
\label{assump:convex}
For $\forall k\in[K]$, $F^{(d)}_k(\varphi)$ is $\mu$-strongly convex, i.e., $\exists \mu$, $\forall \varphi,\varphi^\prime$, such that
\begin{equation}
    \bigl\|\nabla F^{(d)}_k(\varphi) - \nabla F^{(d)}_k(\varphi^\prime)\bigr\| \ge \mu\bigl\|\varphi - \varphi^\prime\bigr\|
\end{equation}
Consequently, $F^{(d)}(\varphi)$ is also $\mu$-strongly convex.
\end{assumption}

\begin{assumption}[Bounded Gradient Variance~\cite{Li2020On,10.5555/2567709.2567769,stich2018local,10.5555/3327345.3327357,10.1609/aaai.v33i01.33015693}]
\label{assump:bounded_variance}
For $\forall k\in[K]$ and $\forall t\in[T^{(d)}]$, the variance of gradients of $F^{(d)}_k$ is bounded, i.e., $\forall x\in\mathcal{D}_k$, $\exists \sigma_k$ such that
\begin{equation}
    \mathbb{E}\left[\bigl\|\nabla F^{(d)}_k(\varphi_k^t;x)-\nabla F^{(d)}_k(\varphi_k^t) \bigr\|^2\right] \le \sigma_k^2
\end{equation}

\end{assumption}

\begin{assumption}[Bounded Gradient Norm~\cite{Li2020On,decoupled,doi:10.1137/16M1080173,NEURIPS2018_a36b598a,10.5555/2567709.2567769,stich2018local,10.5555/3327345.3327357,10.1609/aaai.v33i01.33015693}]
\label{assump:bounded_norm}
For $\forall k\in[K]$ and $\forall t\in[T^{(d)}]$, the expected squared norm of gradients of $F^{(d)}_k$ and $F^{(s)}$ is uniformly bounded, i.e., $\forall x\in\mathcal{D}_k$, $\exists G$ such that
\begin{equation}
    \mathbb{E}\left[\bigl\|\nabla F^{(d)}_k(\varphi_k^t;x)\bigr\|^2\right] \le G
\end{equation}
and
\begin{equation}
    \mathbb{E}\left[\bigl\|\nabla F^{(s)}(\theta_t)\bigr\|^2\right] \le G
\end{equation}
\end{assumption}

\begin{assumption}[Robbins-Monro Conditions~\cite{Robbins1951ASA,decoupled,doi:10.1137/16M1080173,NEURIPS2018_a36b598a}]
\label{assump:robbin}
The learning rate of the server block $\eta^{(s)}_t$ satisfies $\sum_t\eta^{(s)}_t=\infty$ yet $\sum_t\left(\eta^{(s)}_t\right)^2<\infty$.
\end{assumption}

\subsubsection{Convergence of the Device Block}

Let Assumption~\ref{assump:lsmooth} to Assumption~\ref{assump:bounded_norm} hold, then we have

\begin{theorem}[Device Block Convergence]
\label{theorem:device-converge}
If the learning rate is chosen as $\eta^{(d)}_t=\frac{2}{\mu(\gamma+t)}$, then
\begin{equation}
\begin{split}
    \mathbb{E}&\left[F^{(d)}(\varphi_{T^{(d)}})\right] - F^{(d),*} \le \\
    & \frac{\kappa}{\gamma+T^{(d)}-1}\left(\frac{2B}{\mu}+\frac{\mu\gamma}{2}\mathbb{E}\left[\bigl\|\varphi_0-\varphi^*\bigr\|^2\right]\right)
\end{split}
\end{equation}
where $\kappa=\frac{L}{\mu}$, $\gamma=\max\left\{8\kappa,H\right\}$ and $B=\sum_{k=1}^K\left(\frac{n_k}{n}\right)^2\sigma_k^2+6L\Gamma+8(H-1)^2G$.
\end{theorem}

\begin{proof}
    See Theorem~1 in the literature~\cite{Li2020On}.
\end{proof}

\subsubsection{Convergence of the Server Block}

Let Assumption~\ref{assump:lsmooth}, Assumption~\ref{assump:bounded_norm} and Assumption~\ref{assump:robbin} hold, then we have

\begin{theorem}[Server Block Convergence]
\label{theorem:server-converge}
Each term of the following equation converges:
\begin{equation}
\label{eq:server-converge}
\sum_{t=0}^{T^{(s)}}\eta^{(s)}_t\mathbb{E}\left[\bigl\|\nabla F^{(s)}(\theta_t)\bigr\|^2\right] \le \mathbb{E}\left[F^{(s)}(\theta_0)\right]+\frac{GL}{2}\sum_{t=0}^{T^{(s)}}\left(\eta^{(s)}_t\right)^2
\end{equation}
\end{theorem}

\begin{proof}
From Proposition~3.1 in the literature~\cite{decoupled}, we have
\begin{equation}
\label{eq:decoupled}
    \begin{split}
        \sum_{t=0}^{T^{(s)}}\eta^{(s)}_t\mathbb{E}&\left[\bigl\|\nabla F^{(s)}(\theta_t)\bigr\|^2\right] \le \\
        &\mathbb{E}\left[F^{(s)}(\theta_0)\right]+G\sum_{t=0}^{T^{(s)}}\eta^{(s)}_t\left(\sqrt{2c_t}+\frac{L\eta^{(s)}}{2}\right)
    \end{split}
\end{equation}
where $c_t$ is the distance defined in Equation~\labelcref{eq:divergence} at the iteration $t$. Since the server block starts training after the device block has converged, we have $c_t\equiv0$. By substituting $c_t\equiv0$ into Equation~\labelcref{eq:decoupled}, we obtain Equation~\labelcref{eq:server-converge}. From Assumption~\ref{assump:robbin}, the right term is bounded. Therefore, the server block converges as noted in the literature~\cite{doi:10.1137/16M1080173,NEURIPS2018_a36b598a}.
\end{proof}

\subsection{Communication Cost Analysis}
\label{subsec:comm_analysis}

We further analyze the communication volume of \MethodName\ compared to FL and SFL. Let $s^{(d)}$ and $s^{(s)}$ denote the sizes of the device and server blocks, respectively, with $s = s^{(d)} + s^{(s)}$ representing the complete model size. The auxiliary network size is denoted as $s^{(aux)}$, where $s^{(aux)} \ll s^{(s)}$. The total activation size for the dataset is $s^{(act)}$, which typically satisfies $s^{(act)} \gg s^{(s)}$ when the split point is placed early. Table~\ref{tab:size} lists the model sizes and activation sizes for various deep neural networks with $p=1$.

\begin{table}
    \centering
    \caption{Model sizes and activation sizes (GB) for different neural networks with split point $p=1$ on CIFAR-10.}
    \label{tab:size}
    \begin{tabular}{c|c|c|c|c|c}
        \colorrow
        \hline\textbf{DNN Type} & \textbf{Model} &  \textbf{$s^{(act)}$} & \textbf{$s^{(d)}$} & \textbf{$s^{(aux)}$} & \textbf{$s^{(s)}$} \\
        \hline
        \multirow{2}{*}{CNN} & \makecell{MobileNet-\\Large} & 1.53e-1 & 1.34e-5 & 3.47e-5 & 3.18e-2 \\
        \cline{2-6}
        & VGG-11 & 6.09e-1 & 2.04e-5 & 1.19e-3 & 2.10e-1 \\ \hline
        \multirow{2}{*}{Transformer} & Swin-Tiny & 2.29e-1 & 8.83e-4 & 5.75e-4 & 2.04e-1 \\
        \cline{2-6}
        & ViT-Small & 9.28e-1 & 1.34e-2 & 6.83e-3 & 1.46e-1\\
        \hline
    \end{tabular}
\end{table}

By adapting Equation~\labelcref{eq:comm}, the total communication volume of \MethodName\ is:
\begin{equation}
    \rev{C_{Ampere}} = 2N \left(s^{(d)}+s^{(aux)}\right) + s^{(act)}
\end{equation}

For classic SFL, communication consists of both (a) device block exchanges and (b) per-iteration activation and gradient transfers:
\begin{equation}
    C_{SFL} = 2N \left(s^{(d)} + s^{(act)}\right)
\end{equation}

The difference in communication volume achieved by \MethodName\ compared to SFL is:
\begin{equation}
    C_{SFL}-\rev{C_{Ampere}} = (2N-1)s^{(act)}-2Ns^{(aux)}
\end{equation}

Since $s^{(act)} \gg s^{(s)} \gg s^{(aux)}$, this difference is positive, proving that \MethodName\ reduces communication overhead compared to SFL.

FL communication exchanges the full model per round:
\begin{equation}
    C_{FL} = 2Ns=2N \left(s^{(d)} + s^{(s)}\right)
\end{equation}

The difference between FL and \MethodName\ is:
\begin{equation}
    C_{FL}-\rev{C_{Ampere}} = 2N(s^{(s)}-s^{(aux)})-s^{(act)}
\end{equation}

For all models listed in Table~\ref{tab:size}, $C_{FL}-\rev{C_{Ampere}}>0$ holds if training takes $N\ge3$ epochs. This demonstrates that \MethodName\ achieves lower communication overhead than FL. Moreover, as the number of training epochs $N$ increases, the difference of the communication costs between FL and \MethodName\ grows linearly.

\section{Evaluation}
\label{sec:results}
This section evaluates \MethodName\ against relevant baselines. Section~\ref{subsec:setup} describes the experimental setup. Section~\ref{subsec:acc-eff} reports accuracy and efficiency, highlighting that \MethodName\ achieves higher model accuracy within less training time and reduces communication overhead while lowering on-device computation compared to state-of-the-art SFL systems. Section~\ref{subsec:ablation} presents ablation studies that isolate the contribution of each design choice.

\subsection{Experimental Setup}
\label{subsec:setup}

\begin{table*}[t]
    \centering
    \caption{Configuration of the testbed.}
    \label{tab:testbed}
    \begin{tabular}{c|c|c|c|c|c|c}
        \hline
        \colorrow
        \textbf{Device Tier} &\textbf{\#Devices} & \textbf{CPU} & \textbf{CPU Frequency} & \textbf{GPU} & \textbf{GPU Frequency} & \textbf{Memory} \\
        \hline
        \textbf{Tier 1} & 40 & \multirow{3}{*}{\makecell{Quad-Core ARM\\Cortex-A57}} & \multirow{3}{*}{1.7\,GHz} & NVIDIA GM20B  & 921\,MHz & \multirow{3}{*}{4\,GB Unified} \\
        \textbf{Tier 2} & 40 &  &  & NVIDIA GM20B  & 640\,MHz & \\
        \textbf{Tier 3} & 40 &  &  & NVIDIA GM20B  & 320\,MHz & \\
        \hline
        \textbf{Server} & 1 & \makecell{AMD EPYC 7713P\\(128 Cores)} & 2.0\,GHz & NVIDIA A6000 & 1410 MHz & \makecell{256\,GB RAM;\\48\,GB VRAM} \\
        \hline
    \end{tabular}
\end{table*}

This section presents the testbed, baselines, models and datasets used to evaluate \MethodName.

\textbf{Testbed}: \MethodName\ is evaluated on a heterogeneous platform comprising a cluster of 120 Jetson Nano devices (Table~\ref{tab:testbed}), each with 4~GB of unified memory. We create hardware heterogeneity to represent real-world scenarios by allowing devices to operate at three distinct GPU frequencies: 921 MHz (40 devices), 640 MHz (40 devices), and 320 MHz (40 devices). The server has a 128-core AMD EPYC 7713P CPU, \rev{256 GB} of RAM, and an NVIDIA A6000 GPU with 48 GB of VRAM. Unless otherwise specified, the device-server bandwidth is set to 50 Mbps, a network akin to an edge–cloud system~\cite{8567673}. During each training round, 12 devices are randomly selected to participate, each training on 10,000 local samples before transferring its updated model to the server for aggregation.

\textbf{Models and Tasks}: \rev{To validate \MethodName\ across both architecture families and task types, we organize the evaluation into four groups (architecture $\times$ task):}
\rev{\begin{itemize}[leftmargin=*]
    \item \textit{Group A: CNN on image classification} -- VGG-11~\cite{vgg} and MobileNetV3-Large (MobileNet-L)~\cite{mobilenet} on CIFAR-10~\cite{cifar10,cifar10-2} and Tiny ImageNet~\cite{tinyimagenet}.
    \item \textit{Group B: CNN on semantic segmentation} -- FCN~\cite{long2015fully} with a ResNet18~\cite{resnet} backbone on Cityscapes~\cite{cityscapes} and LIP~\cite{lip}.
    \item \textit{Group C: Transformer on image classification} -- Swin-T~\cite{swin} and Vision Transformer-Small (ViT-S)~\cite{steiner2021augreg,vit}\footnote{We used the ViT-S model~\cite{steiner2021augreg}, which is based on the original Vision Transformer architecture~\cite{vit} but with fewer parameters.} on CIFAR-10 and Tiny ImageNet.
    \item \textit{Group D: Transformer on text classification} -- DistilBERT~\cite{distilbert}, a compact Transformer-based language model, on AG News~\cite{agnews} and Stanford Sentiment Treebank v2 (SST-2)~\cite{socher2013sst,wang2018glue}.
\end{itemize}}
\rev{Together these groups span both CNN and Transformer architecture families across diverse image classification, semantic segmentation, and text classification tasks, providing a wide-ranging assessment of \MethodName.}

\textbf{Datasets}: \rev{For image classification (Group~A and Group~C) we use} CIFAR-10 and Tiny ImageNet. CIFAR-10 contains 60,000 images (50,000 for training and 10,000 for testing) with 32$\times$32 pixel resolution, spanning 10 classes. Tiny ImageNet, a subset of ImageNet, includes 110,000 images resized to 64$\times$64 pixels, with 500 training and 50 validation samples per class across 200 classes.
\rev{For semantic segmentation (Group~B) we use Cityscapes, an urban street-scene benchmark, combining its finely and coarsely annotated training images (22,973 training images in total) with the finely annotated validation set (500 images) and evaluating over the 19 standard Cityscapes classes, and LIP, a human-parsing benchmark (30,462 training and 10,000 validation images, 20 classes covering 19 human-part labels plus background). Segmentation quality is reported as the mean Intersection-over-Union (mIoU) over all classes rather than classification accuracy.}
\rev{For text classification (Group~D) we use AG News, a topic-classification benchmark with 120,000 training and 7,600 test samples across four news categories (World, Sports, Business, Sci/Tech), and SST-2, a binary sentiment-classification benchmark with 48,000 training and 19,349 test samples.}

The datasets are partitioned across devices in a non-IID manner using a Dirichlet distribution $Dir(\frac{\alpha}{1-\alpha+\epsilon})$ with $\alpha\in(0,1]$ and a small constant $\epsilon<1e-8$, following prior studies~\cite{acar2021federated,fedavg,fedavgm}. Each device is assigned a class distribution vector sampled from the Dirichlet distribution, determining the probability of each class appearing on that device. Data samples are then allocated by randomly selecting labels based on the probability distribution until all samples are assigned, ensuring a realistic data heterogeneity for FL/SFL. A smaller $\alpha$ increases heterogeneity by concentrating most samples into fewer classes per device, whereas $\alpha=1$ causes $\frac{\alpha}{1-\alpha+\epsilon} \to \infty$, resulting in an approximately IID distribution. By default, we set  $\alpha=0.33$ to represent a moderate non-IID scenario unless stated otherwise. \rev{For semantic segmentation, the same Dirichlet partitioning is applied over the dominant foreground class per image.}

\rev{\textbf{Training Configuration}: All experiments use the SGD optimizer with weight decay $10^{-4}$ and a step-decay learning-rate schedule (decay factor $0.1$ applied every 10 epochs without validation improvement). Training stops early after 30 consecutive aggregation rounds without validation-accuracy improvement. Unless varied in Section~\ref{subsubsec:split_ablation}, the split point is $p=1$. The per-group batch size and learning rate are listed in Table~\ref{tab:training}.}

\begin{table}[ht]
    \centering
    \caption{\rev{Per-group training hyperparameters. Group definitions are provided in Section~\ref{subsec:setup}.}}
    \label{tab:training}
    \begin{tabular}{c|cc}
        \hline
        \colorrow
        \textbf{Group} & \textbf{Batch size} & \textbf{Learning rate} \\
        \hline
        A & 32 & $10^{-2}$ \\
        B & 4  & $10^{-2}$ \\
        C & 32 & $10^{-3}$ \\
        D & 32 & $10^{-3}$ \\
        \hline
    \end{tabular}
\end{table}

\textbf{Baselines}: We compare \MethodName\ against the following SFL baselines (see Section~\ref{sec:rw} for details): (i) SplitFed~\cite{splitfed}, the first SFL system, (ii) PiPar~\cite{pipar}, which uses pipeline parallelism to overlap the communication and computation of SFL, reducing the overall training time, (iii) SplitGP~\cite{splitgp}, which introduces local loss functions on devices. The model learns personalized features from local loss and general features from global loss, improving overall accuracy, and (iv) SplitFed + SCAFFOLD~\cite{scaffold}, where SCAFFOLD mitigates the non-IID issue in FL by maintaining correction terms that reduce the divergence between each device’s model, which is extended to SFL in this paper. \rev{The same split point is used for \MethodName\ and the above baselines.}

\textit{Rationale}: PiPar is selected as representative of systems that improve communication efficiency. Although the communication volume is not reduced in PiPar, communication is overlapped with computation to reduce the overall training time. A variant of SplitGP that relies solely on local loss reduces communication by half but negatively impacts model accuracy. Therefore, we use the version of SplitGP that incorporates both local and global losses as a baseline. SCAFFOLD mitigates non-IID issues in FL. In our implementation of the baseline system, we extend its application to SFL by combining it with SplitFed. \rev{As an orthogonal optimizer-level correction, SCAFFOLD could be layered on any baseline; we pair it with vanilla SplitFed to isolate its variance-reduction effect on a canonical backbone. We omit PiPar~+~SCAFFOLD because PiPar only reschedules communication and computation, leaving the optimization and converged accuracy identical to SplitFed, and omit SplitGP~+~SCAFFOLD because SplitGP already handles heterogeneity via its local and global losses, so stacking the two would conflate two non-IID remedies rather than provide an independent baseline.} Classic FL is not included because all evaluated architectures \rev{require more memory} than the 4~GB unified memory available on the Jetson Nano devices, which makes full-model on-device training infeasible; a network-splitting strategy such as SFL is required. \rev{In Section~\ref{subsubsec:comm}, the FL system's communication cost is estimated and compared to \MethodName.}

\subsection{Accuracy and Efficiency}
\label{subsec:acc-eff}

\MethodName\ is evaluated against the four SFL-based baselines, focusing on model accuracy, communication and computation efficiency, and robustness against non-IID data. The experiments demonstrate that \MethodName\ achieves higher model accuracy with less training time, significantly reduces device-server communication, and effectively mitigates the impact of data heterogeneity.

\begin{table*}[ht]
    \centering
    \caption{\rev{Total number of epochs for training to convergence (lower is better), using \MethodName\ and four SFL baselines, organized by the four (architecture~$\times$~task) groups. For SFL baselines, both device and server blocks are trained in each epoch; in \MethodName\ they are trained sequentially and are therefore reported separately as ``device'' and ``server'' epochs. Total training time is reported in Figure~\ref{fig:time2acc:all}.}}
    \label{tab:epoch}
    \begin{tabular}{c|l|cccc|cc}
        \hline
        \colorrow
        \multicolumn{2}{c|}{\textbf{\rev{Experiment}}} & \multicolumn{4}{c|}{\textbf{SFL Baselines}} & \multicolumn{2}{c}{\textbf{\MethodName}} \\
        \hhline{~~|------}
        \colorrow
        \textbf{\rev{Group}} & \textbf{\rev{(Model, Dataset)}} & \footnotesize SplitFed & \footnotesize PiPar & \footnotesize SCAFFOLD & \footnotesize SplitGP & \footnotesize \textit{device} & \footnotesize \textit{server} \\
        \hline
        \multirow{4}{*}{\textbf{\rev{A}}}
            & \rev{MobileNet-L, CIFAR-10}        & 200 & 210 & 240 & 300 & \textbf{55}  & 32 \\
            & \rev{MobileNet-L, Tiny ImageNet}   & 159 & 186 & 258 & 298 & \textbf{36}  & 10 \\
            & \rev{VGG-11, CIFAR-10}              & 115 & 121 & 184 & 211 & \textbf{61}  & 25 \\
            & \rev{VGG-11, Tiny ImageNet}         & 109 & 117 & 156 & 148 & \textbf{78}  & 24 \\
        \hline
        \multirow{2}{*}{\textbf{\rev{B}}}
            & \rev{FCN-ResNet18, Cityscapes}
            & \rev{93} & \rev{223} & \rev{144} & \rev{121} & \rev{\textbf{67}} & \rev{19} \\
            & \rev{FCN-ResNet18, LIP}
            & \rev{431} & \rev{334} & \rev{480} & \rev{211} & \rev{\textbf{46}} & \rev{14} \\
        \hline
        \multirow{4}{*}{\textbf{\rev{C}}}
            & \rev{Swin-T, CIFAR-10}              & 120 & 152 & 216 & 240 & \textbf{55}  & 22 \\
            & \rev{Swin-T, Tiny ImageNet}         & 218 & 229 & 223 & 240 & \textbf{127} & 21 \\
            & \rev{ViT-S, CIFAR-10}                & 131 & 135 & 244 & 201 & \textbf{81}  & 46 \\
            & \rev{ViT-S, Tiny ImageNet}           & 105 & 107 & 110 & 115 & \textbf{54}  & 88 \\
        \hline
        \multirow{2}{*}{\textbf{\rev{D}}}
            & \rev{DistilBERT, AG News}          & \rev{195} & \rev{108} & \rev{93}  & \rev{52}  & \rev{\textbf{68}} & \rev{30} \\
            & \rev{DistilBERT, SST-2}             & \rev{161} & \rev{186} & \rev{220} & \rev{161} & \rev{\textbf{77}} & \rev{52} \\
        \hline
    \end{tabular}
\end{table*}

\subsubsection{Model Accuracy}
\label{subsubsec:acc}

We evaluate the model accuracy of \MethodName\ against the SFL-based baselines to highlight its ability to achieve higher accuracy with less training time. \rev{A model is considered converged if its validation accuracy does not improve for 30 consecutive aggregation rounds, at which point training is halted by early stopping~\cite{earlystop} (see Section~\ref{subsec:setup} for the full schedule). This single criterion is applied identically to every method to ensure a fair comparison, and for \MethodName\ it is evaluated separately on the device block (determining $N^{(d)}$) and the server block (determining $N^{(s)}$).}

Figure~\ref{fig:time2acc:all} shows model accuracy against training time. In \MethodName, model accuracy increases during initial device-block training. Once the device block converges, server block training begins from scratch, causing a temporary accuracy drop in all cases shown in the figure. However, accuracy quickly recovers and surpasses the highest accuracy achieved during device block training. This confirms that final accuracy is improved from server-side training beyond what the auxiliary network alone can achieve.

Compared to SFL baselines that train the device and server blocks together, \MethodName\ converges faster on the device block and reaches higher final accuracy after server-block training, attaining the highest final accuracy in the shortest training time. Because the device block and auxiliary network are far smaller than the full model, they require fewer epochs to converge (Table~\ref{tab:epoch}), and devices can exit once the device block converges rather than remaining active throughout, allowing resource-constrained devices to finish sooner and free up for other tasks. Across all groups, \MethodName\ achieves a training speedup of up to \rev{18.6$\times$} relative to the fastest baseline and matches or improves the best baseline's final accuracy on every benchmark, with one exception where it trails by 0.13 percentage points. \rev{Group~A (MobileNet-L, VGG-11) improves final accuracy by 3.63 to 11.70 percentage points over the strongest baseline at a 3.5$\times$--18.6$\times$ speedup.} \rev{Group~B (FCN-ResNet18) is a dense-prediction task measured by mIoU rather than classification accuracy; \MethodName\ delivers its largest margins here, reaching 28.42\% mIoU on Cityscapes ($+6.81$ percentage points over the best baseline) and 23.94\% on LIP ($+10.85$ percentage points, a 1.8$\times$ relative improvement) at a 2.1$\times$ to 9.5$\times$ speedup over the fastest baseline.} \rev{Group~C (Swin-T, ViT-S) improves final accuracy by 2.95 to 7.01 percentage points at a 1.9$\times$--2.4$\times$ speedup.} \rev{Group~D (DistilBERT) matches the best baseline within 0.3 percentage points (from $-0.13$ to $+0.33$) while training 1.8$\times$--4.5$\times$ faster, as the SFL baselines are already near-saturated on the text benchmarks.}


\begin{figure*}[tp]
    \centering
    \begin{subfigure}[t]{0.24\linewidth}
        \includegraphics[width=0.97\linewidth]{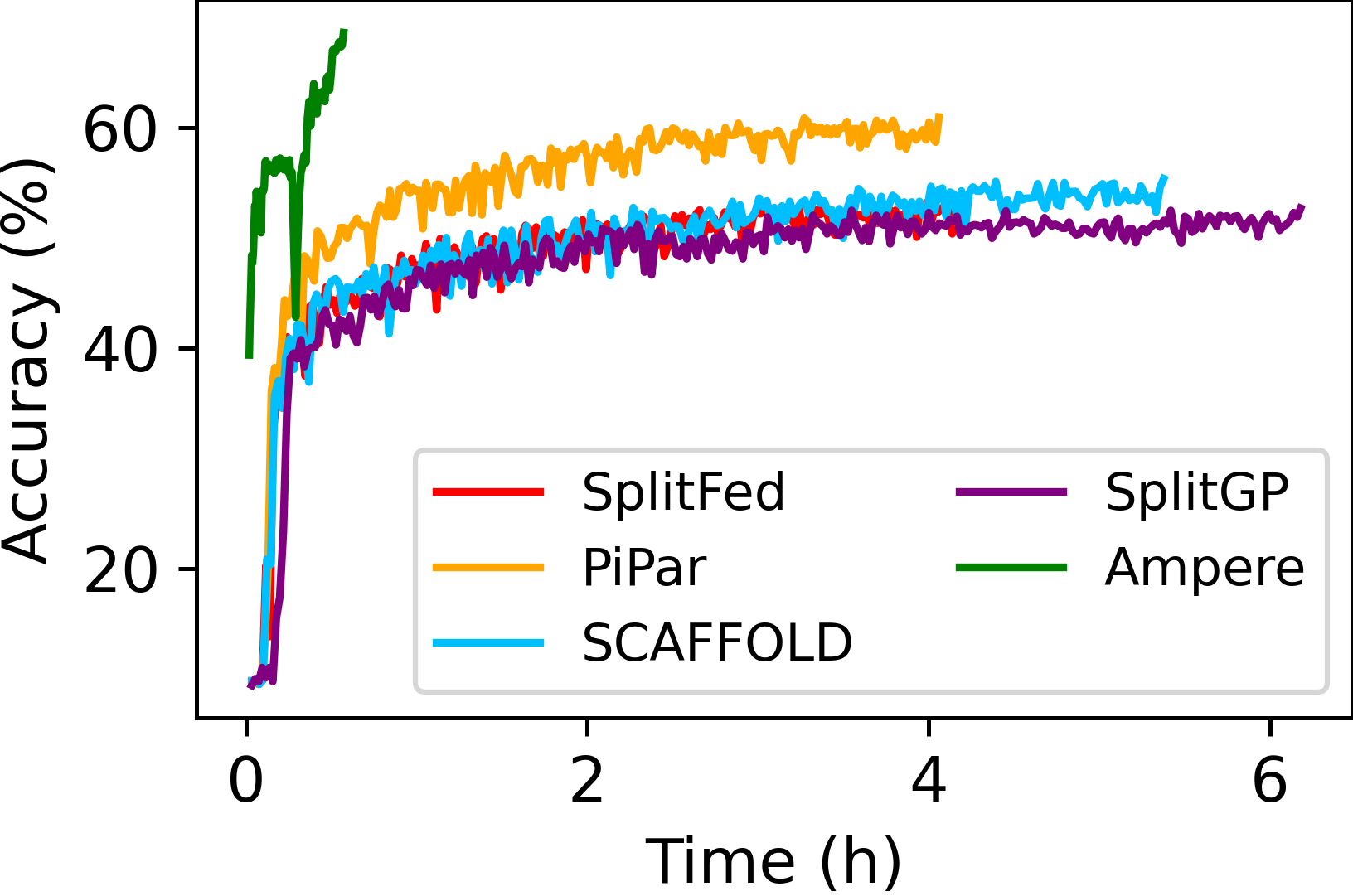}
        \caption{MobileNet-L on CIFAR-10}
        \label{subfig:mobile_cifar10}
    \end{subfigure}\hfill
    \begin{subfigure}[t]{0.24\linewidth}
        \includegraphics[width=0.97\linewidth]{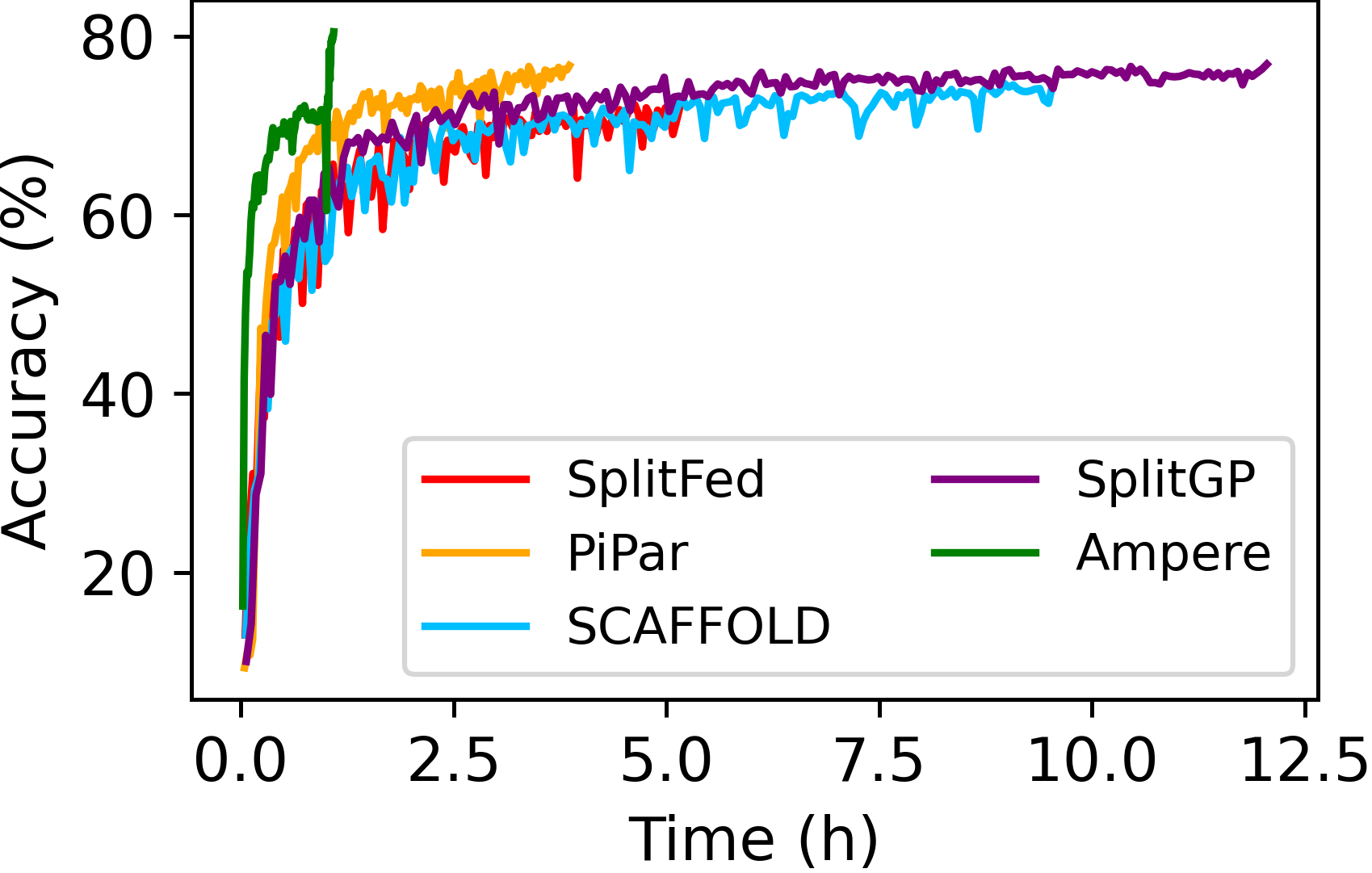}
        \caption{VGG-11 on CIFAR-10}
        \label{subfig:vgg11_cifar10}
    \end{subfigure}\hfill
    \begin{subfigure}[t]{0.24\linewidth}
        \includegraphics[width=0.97\linewidth]{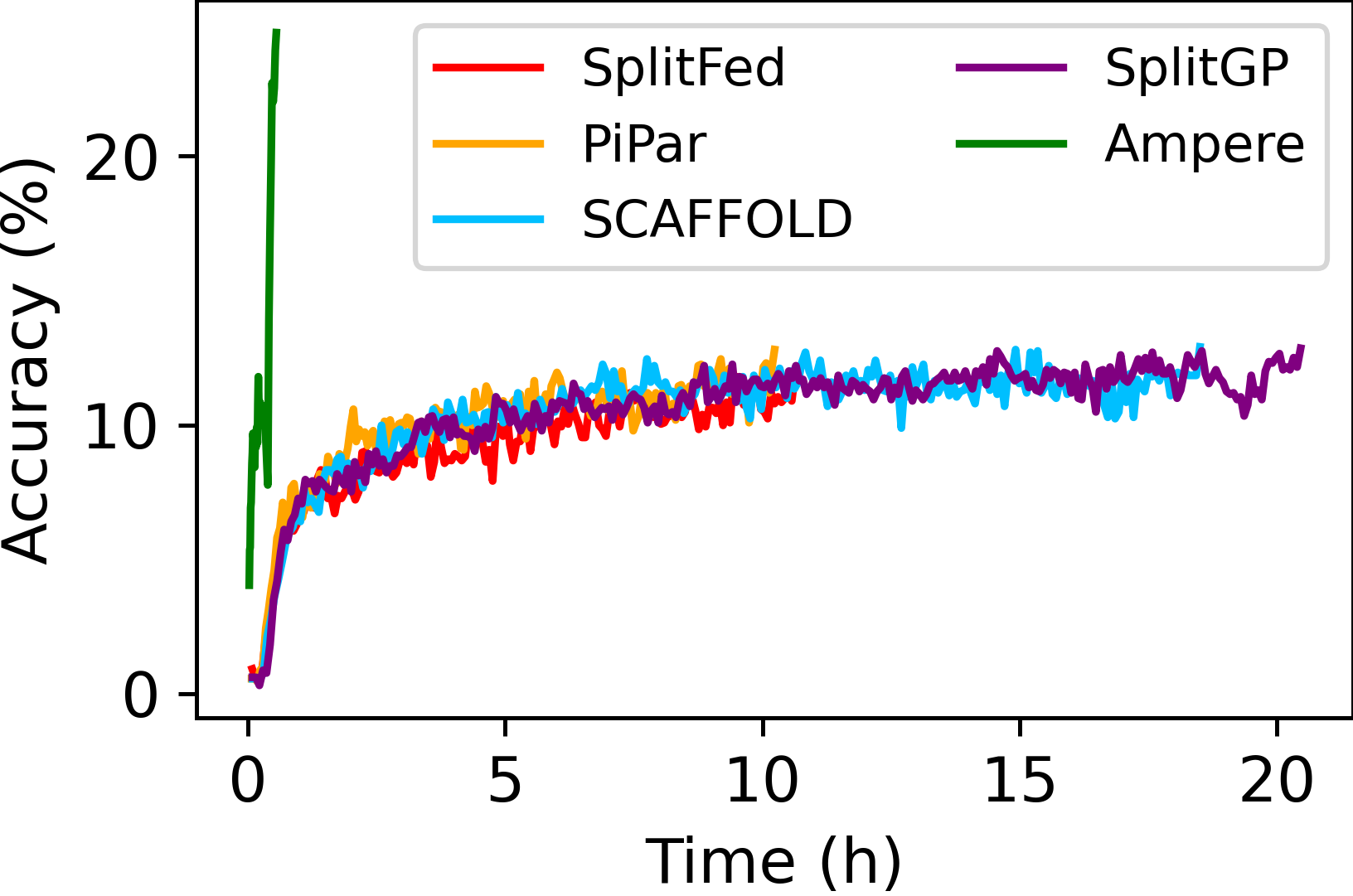}
        \caption{\scriptsize MobileNet-L on Tiny ImageNet}
        \label{subfig:mobile_tinyimagenet}
    \end{subfigure}\hfill
    \begin{subfigure}[t]{0.24\linewidth}
        \includegraphics[width=0.97\linewidth]{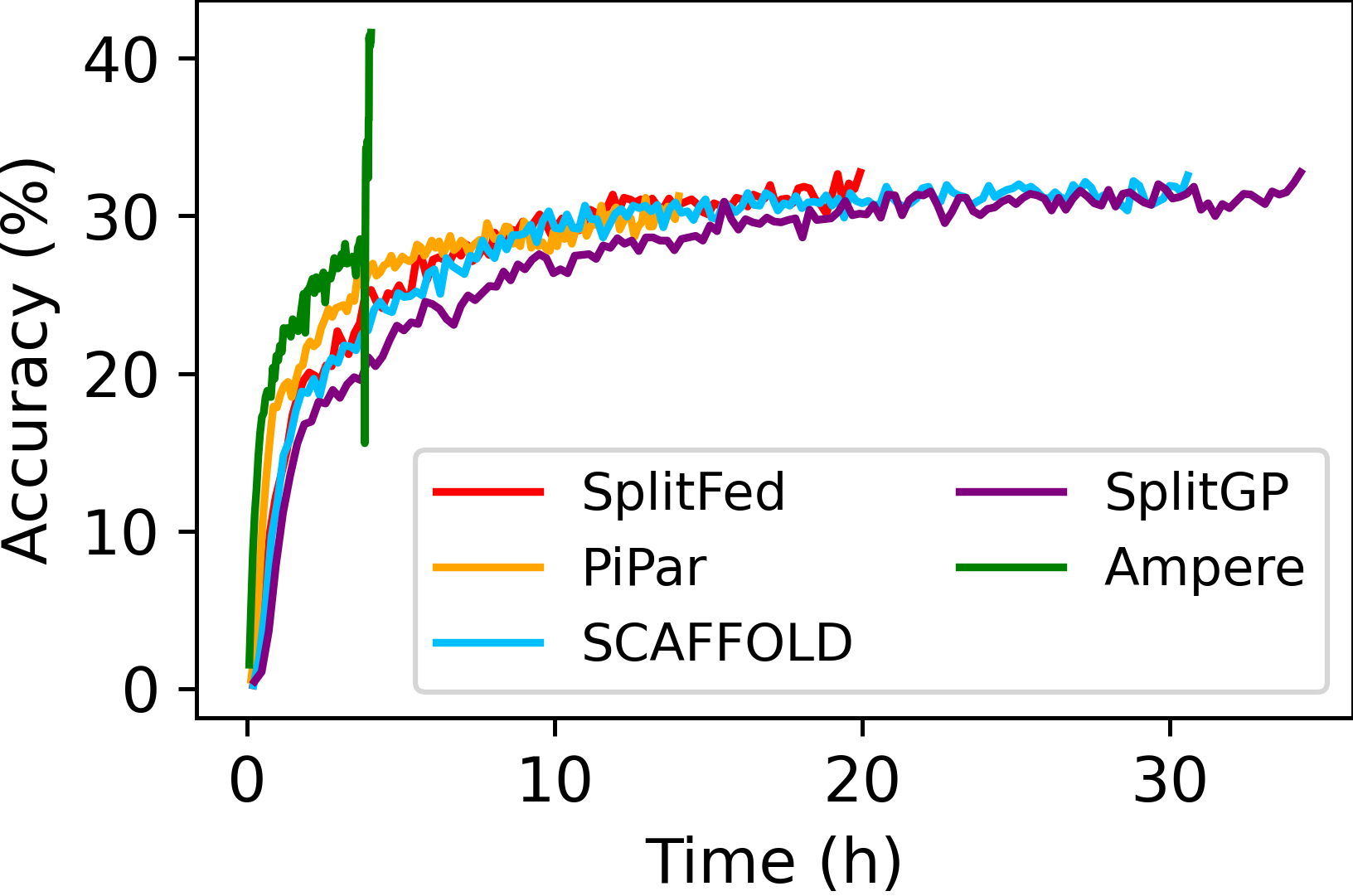}
        \caption{\scriptsize VGG-11 on Tiny ImageNet}
        \label{subfig:vgg11_tinyimagenet}
    \end{subfigure}

    \vspace{6pt}

    \begin{subfigure}[t]{0.24\linewidth}
        \includegraphics[width=0.97\linewidth]{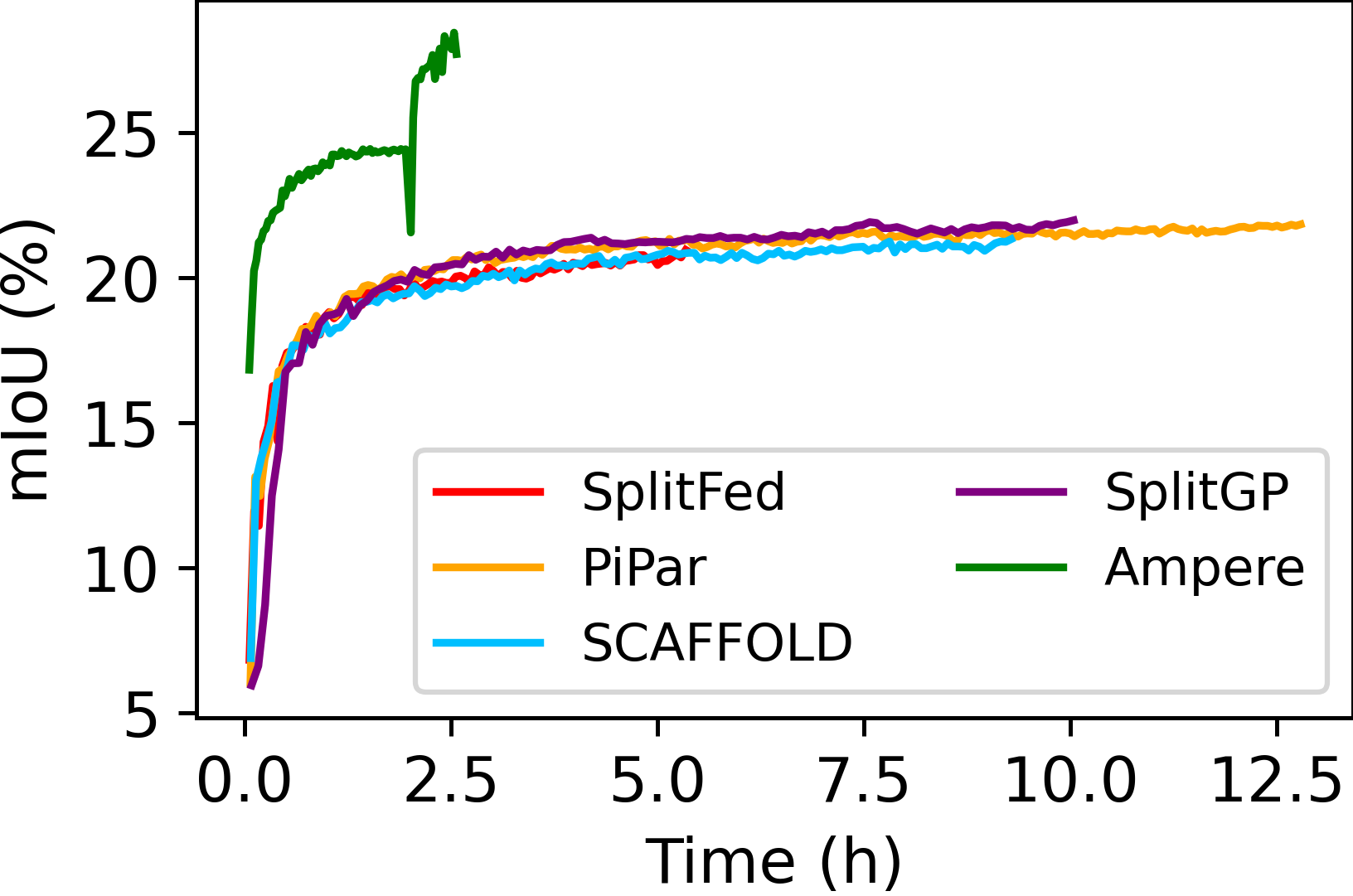}
        \caption{\scriptsize FCN-ResNet18 on Cityscapes}
        \label{subfig:fcn_cityscapes}
    \end{subfigure}\hfill
    \begin{subfigure}[t]{0.24\linewidth}
        \includegraphics[width=0.97\linewidth]{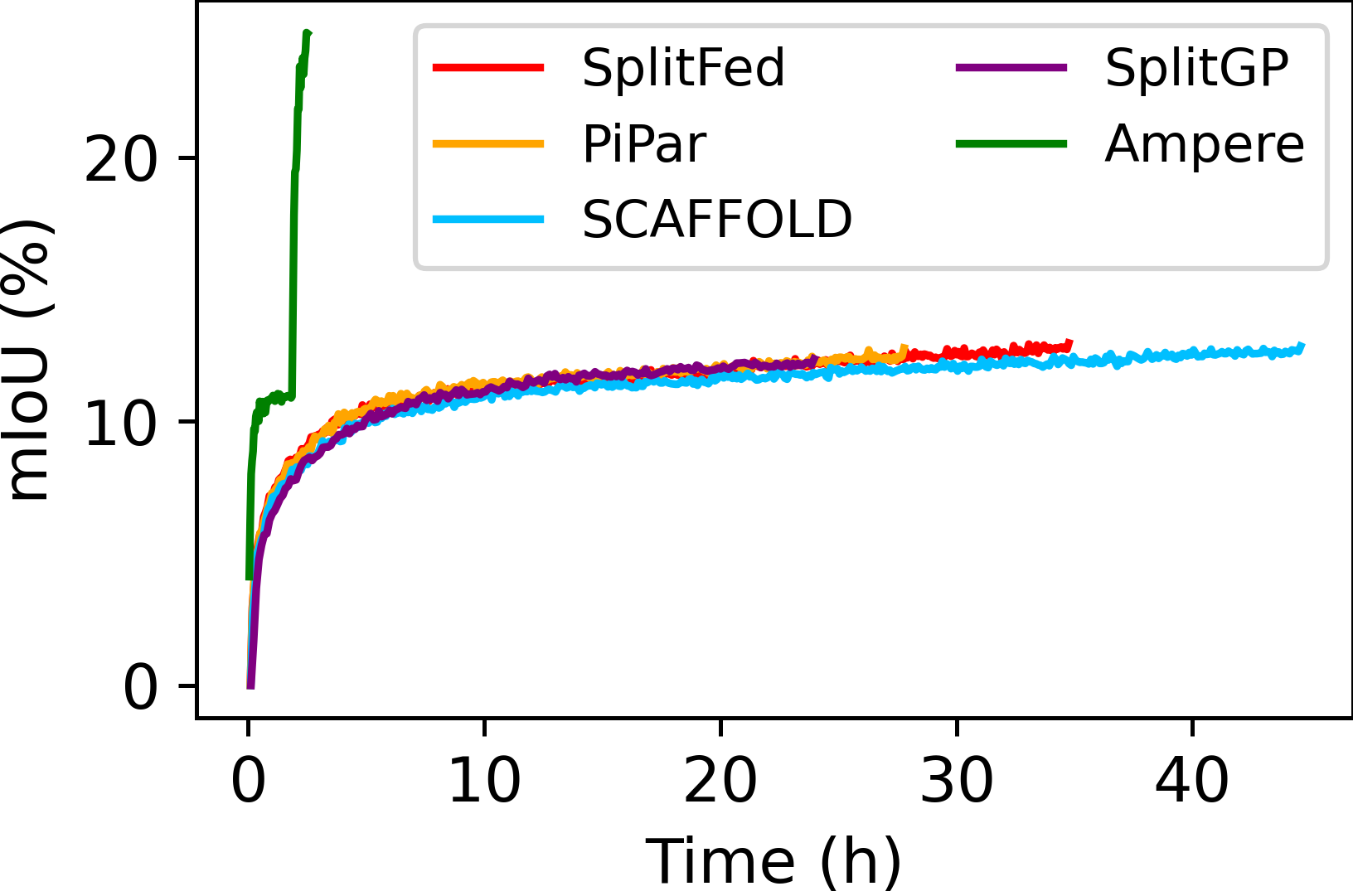}
        \caption{FCN-ResNet18 on LIP}
        \label{subfig:fcn_lip}
    \end{subfigure}\hfill
    \begin{subfigure}[t]{0.24\linewidth}
        \includegraphics[width=0.97\linewidth]{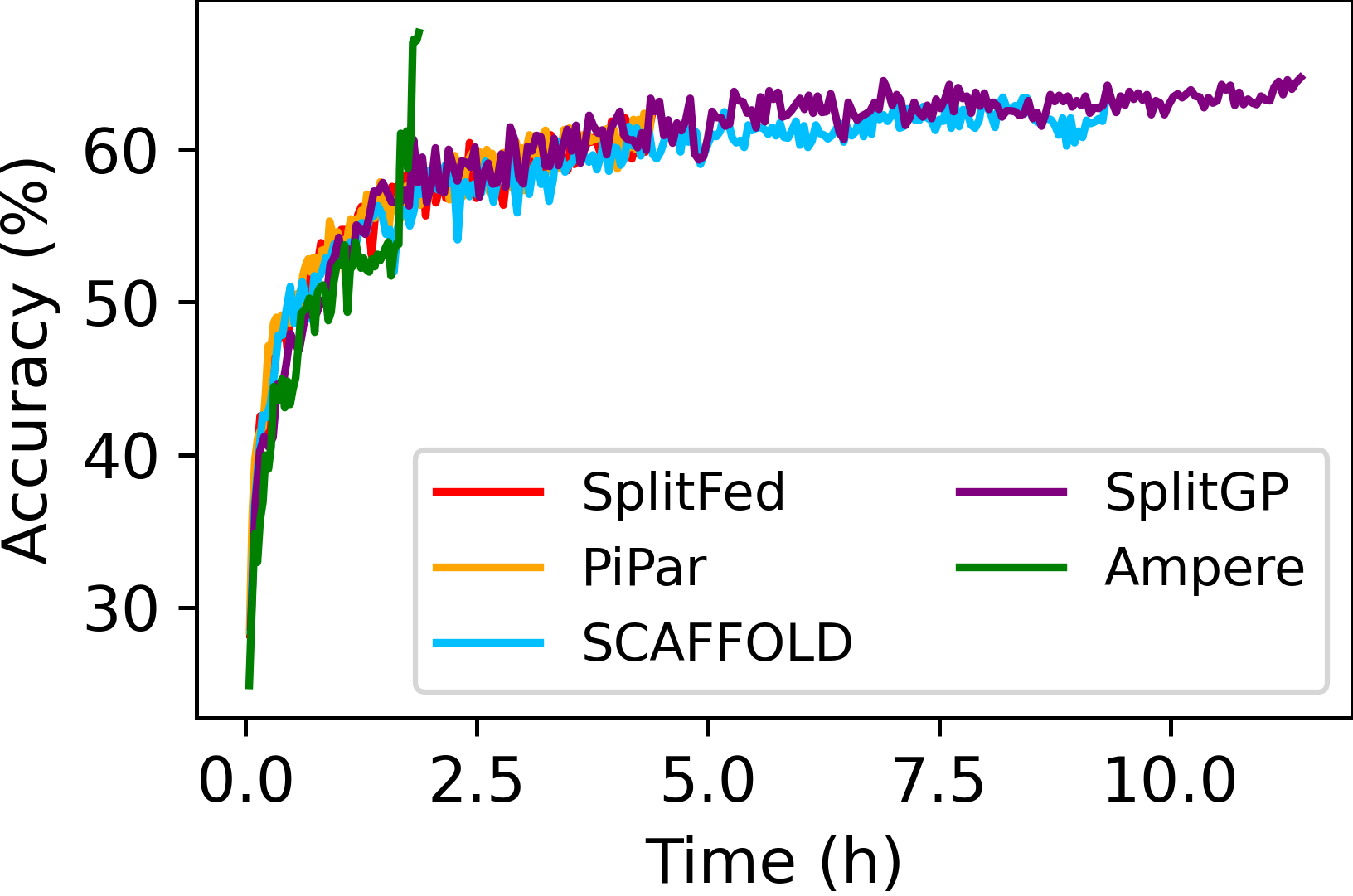}
        \caption{Swin-T on CIFAR-10}
        \label{subfig:swint_cifar10}
    \end{subfigure}\hfill
    \begin{subfigure}[t]{0.24\linewidth}
        \includegraphics[width=0.97\linewidth]{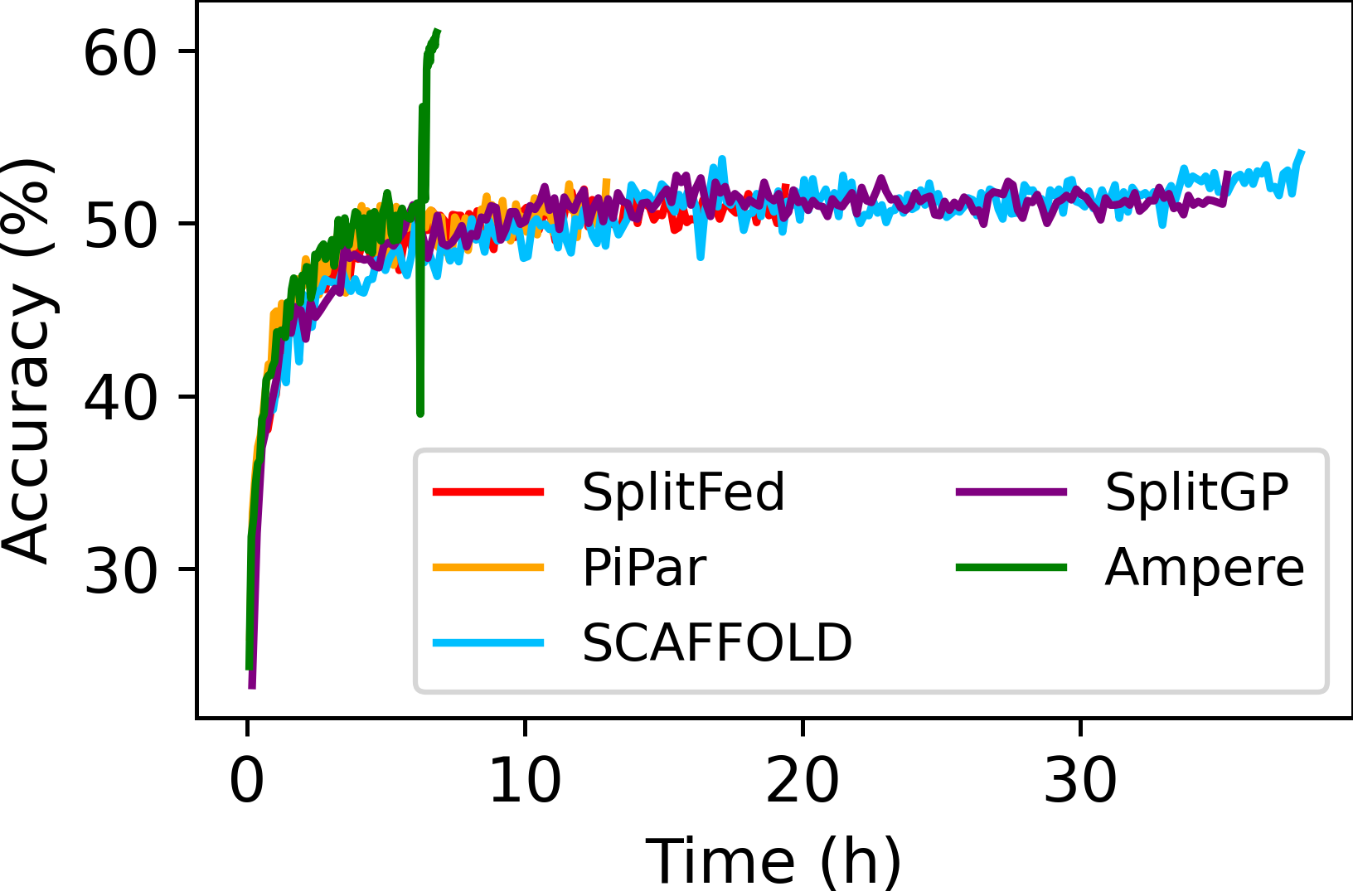}
        \caption{ViT-S on CIFAR-10}
        \label{subfig:vits_cifar10}
    \end{subfigure}

    \vspace{6pt}

    \begin{subfigure}[t]{0.24\linewidth}
        \includegraphics[width=0.97\linewidth]{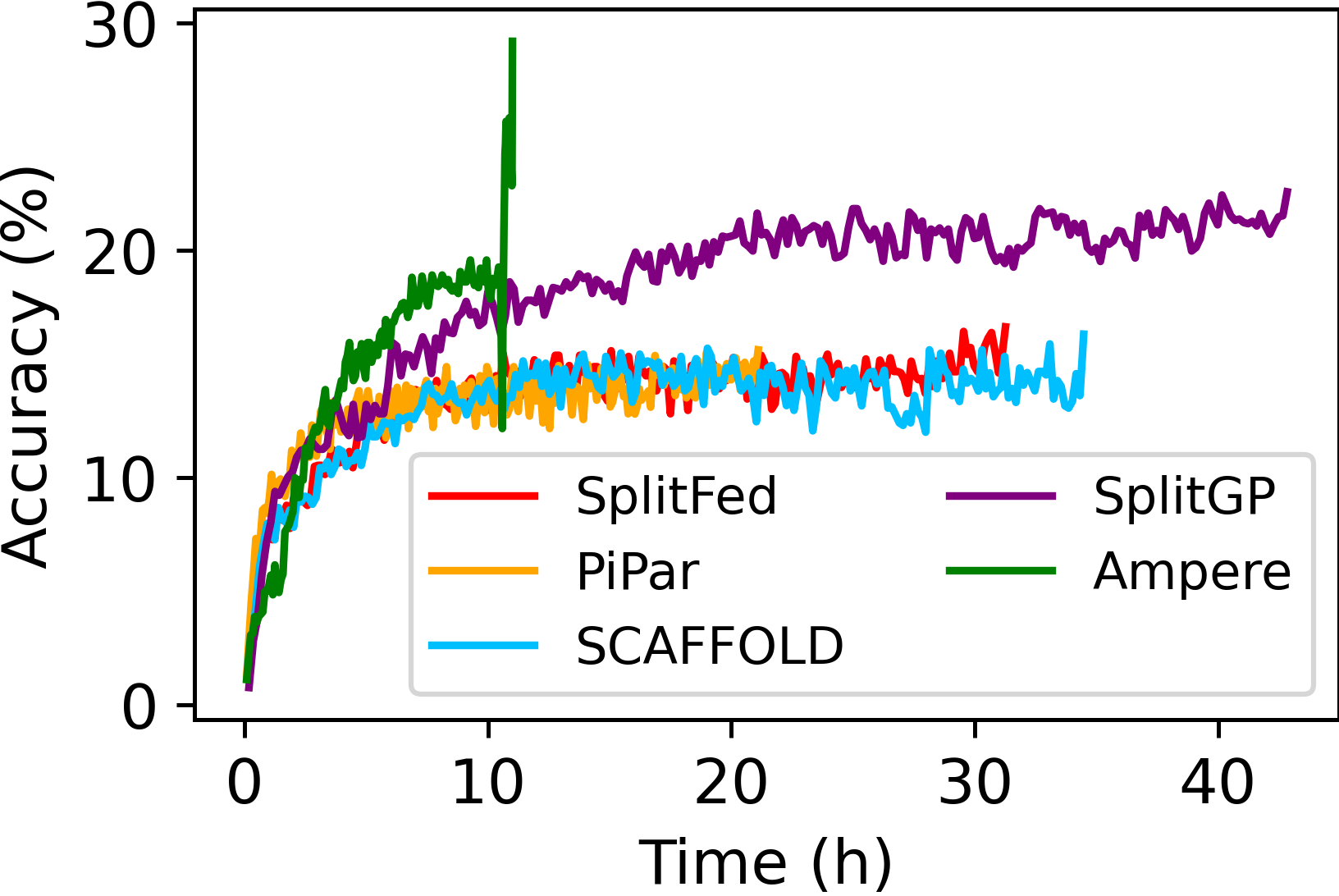}
        \caption{\scriptsize Swin-T on Tiny ImageNet}
        \label{subfig:swint_tinyimagenet}
    \end{subfigure}\hfill
    \begin{subfigure}[t]{0.24\linewidth}
        \includegraphics[width=0.97\linewidth]{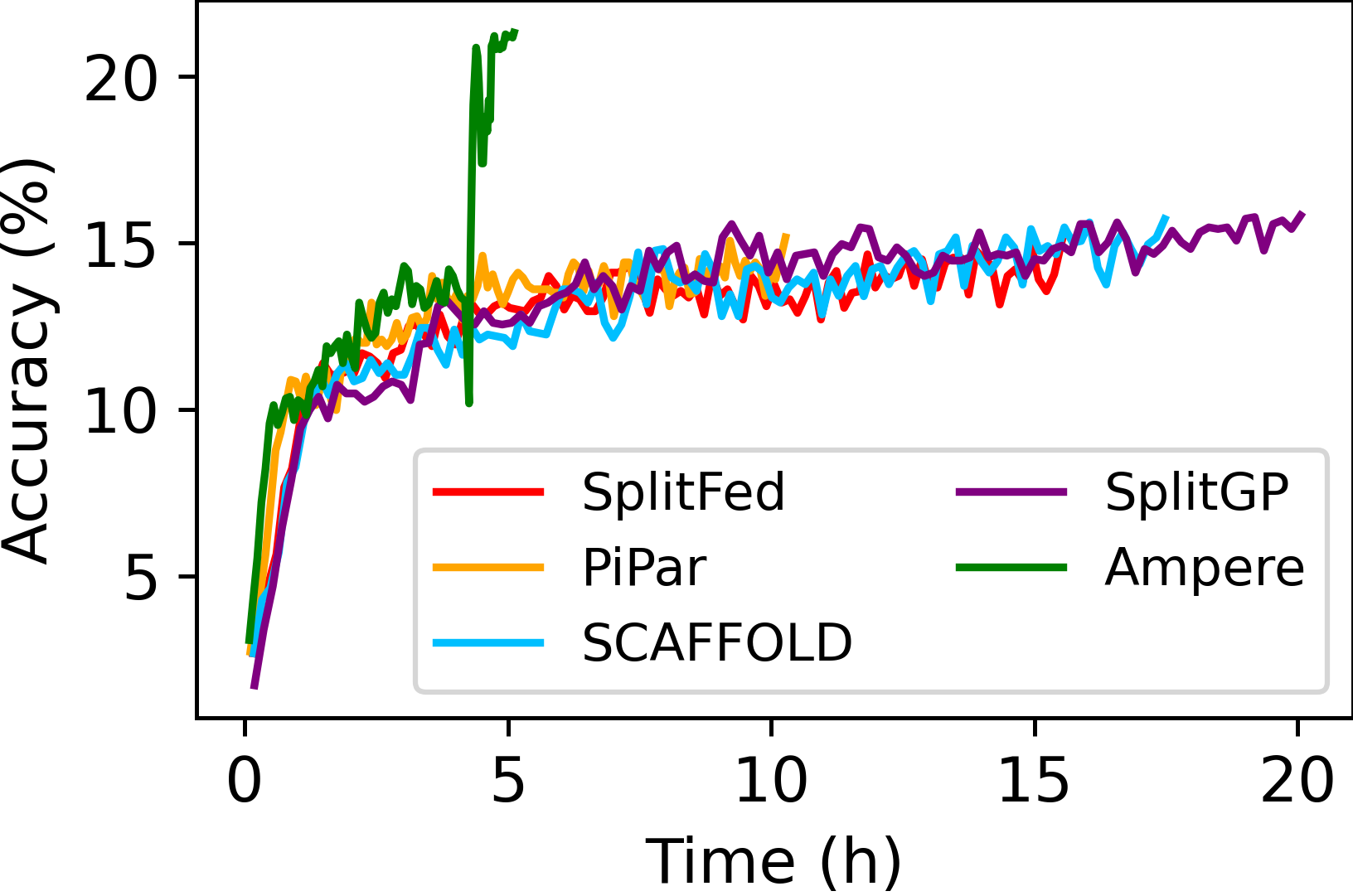}
        \caption{\scriptsize ViT-S on Tiny ImageNet}
        \label{subfig:vits_tinyimagenet}
    \end{subfigure}\hfill
    \begin{subfigure}[t]{0.24\linewidth}
        \includegraphics[width=0.97\linewidth]{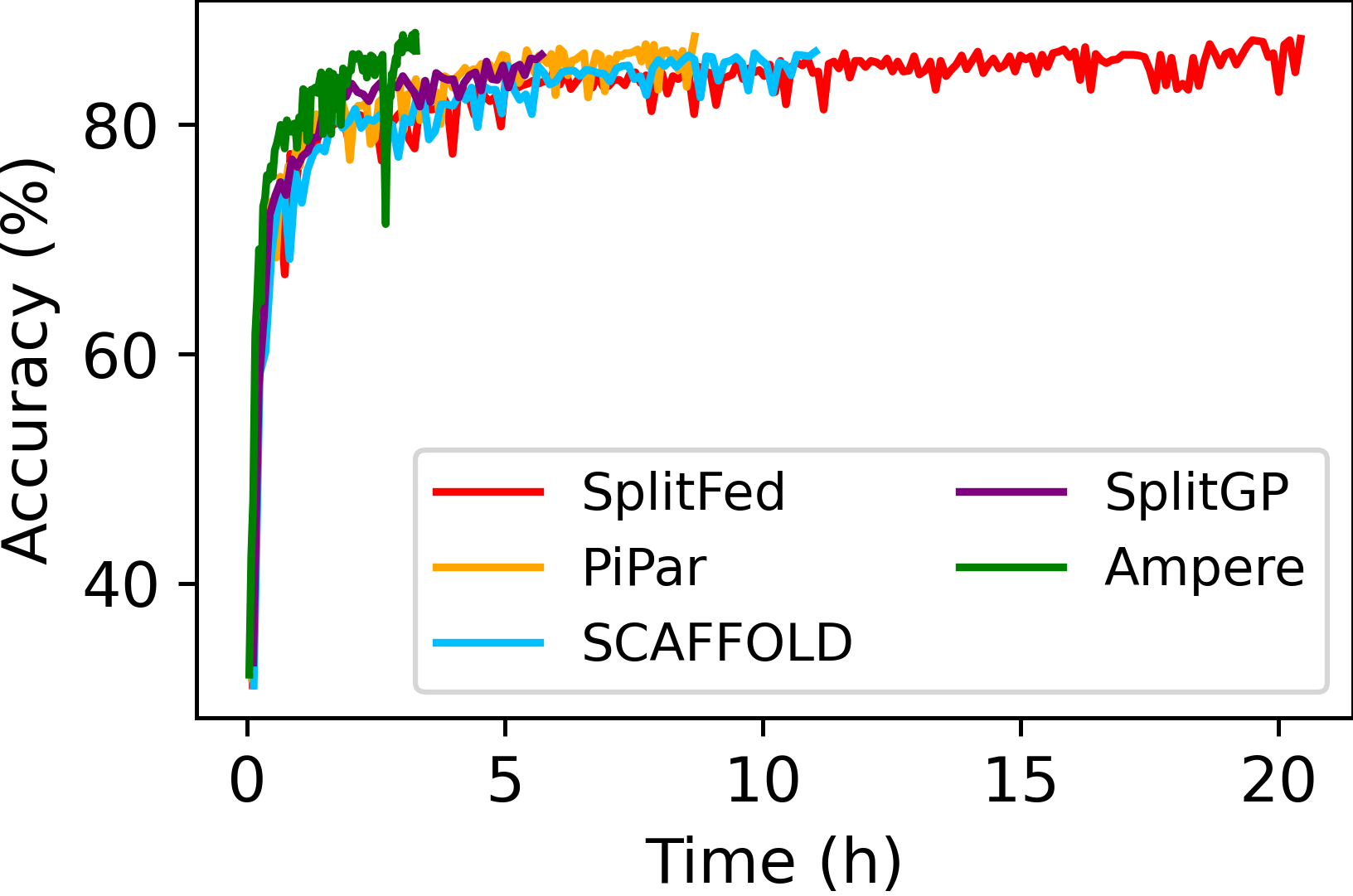}
        \caption{DistilBERT on AG News}
        \label{subfig:distilbert_agnews}
    \end{subfigure}\hfill
    \begin{subfigure}[t]{0.24\linewidth}
        \includegraphics[width=0.97\linewidth]{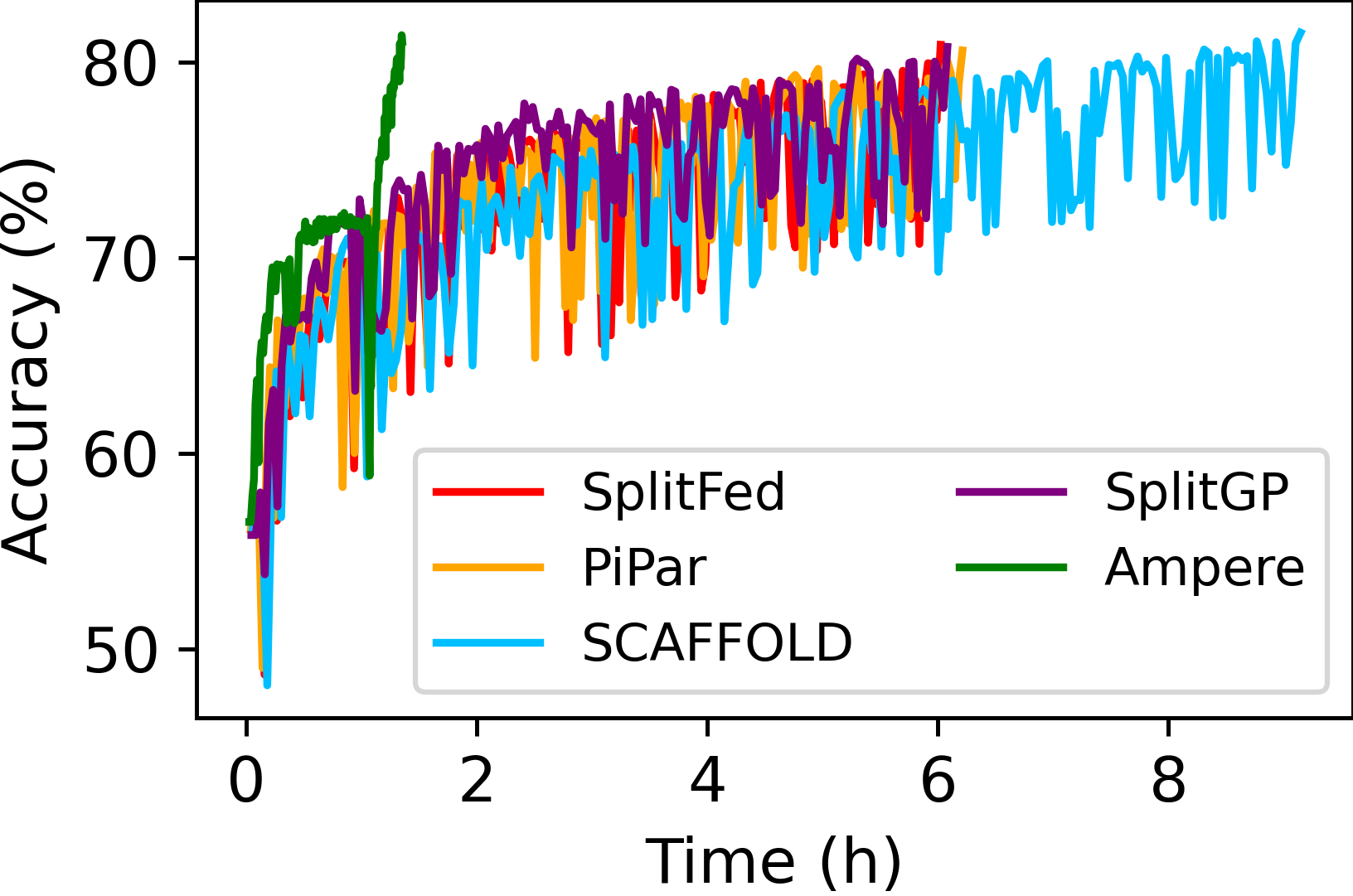}
        \caption{DistilBERT on SST-2}
        \label{subfig:distilbert_sst2}
    \end{subfigure}

    \caption{\rev{Model accuracy (higher is better; mIoU for Group~B) versus training time (lower is better) for all experiment groups. Each curve is drawn up to its best-accuracy step (the convergence point used for the reported epochs and training time); the subsequent early-stopping patience rounds are omitted.}}
    \label{fig:time2acc:all}
\end{figure*}

\begin{tcolorbox}[
    width=0.49\textwidth,
    colframe=black!50!black,
    colback=gray!8,boxrule=0.5pt,
    left=2pt,right=2pt,top=2pt,bottom=2pt]
\textbf{Observation 1:}
\MethodName\ reaches higher final accuracy in less total training time while shortening device participation, across diverse model architectures and task types.
\end{tcolorbox}

\subsubsection{Communication}
\label{subsubsec:comm}

\begin{table*}[ht]
    \centering
    \caption{\rev{Device-server communication volume per device (GB, lower is better) for training to convergence, using FL (estimated), four SFL baselines, and \MethodName, organized by the four (architecture~$\times$~task) groups.}}
    \label{tab:comm}
    \begin{tabular}{c|l|c|cccc|c}
        \hline
        \colorrow
        \multicolumn{2}{c|}{\textbf{\rev{Experiment}}} & \textbf{\rev{FL}} & \multicolumn{4}{c|}{\textbf{SFL Baselines}} & \textbf{\MethodName} \\
        \hhline{~~|~----~}
        \colorrow
        \textbf{\rev{Group}} & \textbf{\rev{(Model, Dataset)}} & \footnotesize \rev{(estimated)} & \footnotesize SplitFed & \footnotesize PiPar & \footnotesize SCAFFOLD & \footnotesize SplitGP & \\
        \hline
        \multirow{4}{*}{\textbf{\rev{A}}}
            & \rev{MobileNet-L, CIFAR-10}        & 12.72 & 61.03  & 64.17  & 73.23  & 91.63  & \rev{\textbf{0.067}} \\
            & \rev{MobileNet-L, Tiny ImageNet}   & 10.05 & 193.85 & 226.85 & 314.55 & 370.60 & \rev{\textbf{0.619}} \\
            & \rev{VGG-11, CIFAR-10}              & 48.30 & 70.20  & 73.91  & 112.42 & 128.92 & \rev{\textbf{0.201}} \\
            & \rev{VGG-11, Tiny ImageNet}         & 42.42 & 265.80 & 285.36 & 380.50 & 368.13 & \rev{\textbf{1.581}} \\
        \hline
        \multirow{2}{*}{\textbf{\rev{B}}}
            & \rev{FCN-ResNet18, Cityscapes}
            & \rev{8.57} & \rev{133.23} & \rev{327.51} & \rev{188.04} & \rev{163.76} & \rev{\textbf{0.889}} \\
            & \rev{FCN-ResNet18, LIP}
            & \rev{39.70} & \rev{676.88} & \rev{641.76} & \rev{732.59} & \rev{355.26} & \rev{\textbf{1.061}} \\
        \hline
        \multirow{4}{*}{\textbf{\rev{C}}}
            & \rev{Swin-T, CIFAR-10}              & 48.55 & 55.00  & 69.74  & 99.19  & 110.00 & \rev{\textbf{0.176}} \\
            & \rev{Swin-T, Tiny ImageNet}         & 88.94 & 398.82 & 419.04 & 408.16 & 439.10 & \rev{\textbf{0.966}} \\
            & \rev{ViT-S, CIFAR-10}                & 38.25 & 245.06 & 252.60 & 459.76 & 376.01 & \rev{\textbf{2.025}} \\
            & \rev{ViT-S, Tiny ImageNet}           & 30.66 & 196.47 & 200.25 & 207.36 & 215.24 & \rev{\textbf{1.918}} \\
        \hline
        \multirow{2}{*}{\textbf{\rev{D}}}
            & \rev{DistilBERT, AG News}          & \rev{25.58} & \rev{391.64} & \rev{241.49} & \rev{229.16} & \rev{145.48} & \rev{\textbf{8.20}} \\
            & \rev{DistilBERT, SST-2}             & \rev{21.12} & \rev{125.74} & \rev{142.41} & \rev{174.19} & \rev{125.74} & \rev{\textbf{3.33}} \\
        \hline
    \end{tabular}
\end{table*}

Table~\ref{tab:comm} presents the per-device volume of data exchanged between each device and the server. By transferring activations only once, \MethodName\ transfers \rev{17.7$\times$ to 911$\times$} less per-device communication volume compared to all SFL baselines.

Although FL is infeasible for these models on the devices due to memory constraints, we estimate its communication volume. Since the only difference between FL and SplitFed is where the computation is carried out, both require a similar number of training epochs. In each epoch, FL transfers two full models for aggregation (one upload and one download per device). We estimate the total communication volume of FL by multiplying the size of the two models by the number of training epochs. The results show that \MethodName\ has \rev{3.1$\times$ to 276$\times$} lower communication than FL across all groups.

Across all completed experiments, \MethodName\ transfers only 0.067--8.20\,GB per device, versus 55--733\,GB for the SFL baselines. The advantage compounds as the workload grows: for VGG-11, the cheapest SFL baseline rises from 70\,GB per device on CIFAR-10 to 266\,GB on Tiny ImageNet, whereas \MethodName\ stays at 0.2--1.6\,GB; at the 50\,Mbps device--server link that 266\,GB transfer alone costs about 12 hours, against a few minutes for \MethodName. The reason is structural: SFL transfers activations and gradients every iteration, so its volume scales with both the epoch count and the dataset size, while \MethodName\ transfers activations once and otherwise exchanges only the small device block and auxiliary network, whose size is independent of the dataset. \rev{Per group, the per-device reduction over the cheapest SFL baseline is 168$\times$--911$\times$ for Group~A.} \rev{For Group~B the gap is widest (150$\times$ on Cityscapes and 335$\times$ on LIP), because semantic segmentation produces large, dense per-pixel activations that the SFL baselines transfer every iteration (on LIP this reaches 642--733\,GB per device for the most communication-intensive baselines), whereas \MethodName\ transfers them once and stays at 0.9--1.1\,GB.} \rev{For Group~C it is 102$\times$--413$\times$.} \rev{For Group~D it is 17.7$\times$--37.8$\times$, the smallest of the four because DistilBERT's token activations are large relative to its model, yet \MethodName\ still reduces per-device communication by more than an order of magnitude relative to the cheapest SFL baseline.}

\input{figures/computation_all}

\begin{tcolorbox}[
    width=0.49\textwidth,
    colframe=black!50!black,
    colback=gray!8,boxrule=0.5pt,
    left=2pt,right=2pt,top=2pt,bottom=2pt]
\textbf{Observation 2:}
\MethodName\ reduces per-device communication below both SFL and FL, by one to three orders of magnitude over SFL.
\end{tcolorbox}

\input{figures/hetero_all}

\subsubsection{On-Device Computation}
\label{subsubsec:comp}

\MethodName\ also incurs the lowest on-device computation, measured in tera floating-point operations (TFLOPs, Figure~\ref{fig:comp:all}), of all methods in every experiment. A primary reason is its faster device-side convergence: Table~\ref{tab:epoch} shows that the device block converges in 36--127 epochs, whereas the SFL baselines keep every device active for up to 480 epochs, and on-device computation increases with the number of device epochs. The remaining saving comes from the per-epoch device workload: the SFL baselines run the same device block, but the more compute-intensive ones add extra per-epoch computation on the device---SCAFFOLD maintains per-device correction terms, and SplitGP computes an additional local loss alongside the global loss---whereas \MethodName\ trains only a shallow device block with a lightweight auxiliary network. \rev{In Group~A, \MethodName\ incurs up to 14.5$\times$ lower on-device computation relative to the most compute-intensive baseline, attained on the largest CNN workload, where the extra per-epoch computation that baseline performs on the device is greatest.} \rev{In Group~B, \MethodName\ incurs up to 10.1$\times$ lower on-device computation relative to the most compute-intensive baseline (on LIP, 9.45 versus 95.82\,TFLOPs). Here the saving comes almost entirely from faster convergence rather than a lighter per-epoch workload: the device block converges in 46--67 epochs, whereas the SFL baselines remain active for 93--480 epochs, and \MethodName\ uses the same single-layer device block as the least compute-intensive baseline, giving both an identical per-epoch device workload.} \rev{In Group~C, \MethodName\ incurs up to 2.6$\times$ lower on-device computation, and in Group~D up to 3.9$\times$ lower, relative to the most compute-intensive baseline.}

\begin{tcolorbox}[
    width=0.49\textwidth,
    colframe=black!50!black,
    colback=gray!8,boxrule=0.5pt,
    left=2pt,right=2pt,top=2pt,bottom=2pt]
\textbf{Observation 3:}
\MethodName\ incurs the lowest on-device computation of all methods due to its faster device-side convergence and smaller per-epoch device workload.
\end{tcolorbox}

\subsubsection{\rev{Robustness to Data Heterogeneity}}
\label{subsubsec:hetero}

The impact of data heterogeneity in \MethodName\ is captured by model accuracy and its standard deviation for different degrees of non-IID data ($\alpha$) relative to SFL baselines. The effect of $\alpha$ is not linear. We consider three representative degrees: IID ($\alpha=1$), moderately non-IID ($\alpha=0.33$), and severely non-IID ($\alpha=0.1$).

Figure~\ref{fig:hetero:all} shows the model accuracy \rev{for Groups~A, C, and D, and mIoU for Group~B} across different non-IID degrees. \MethodName\ consistently outperforms all baselines, achieving up to \rev{12.81 percentage points higher accuracy/mIoU}. We measure the standard deviation \rev{of this metric across the three $\alpha$ values} to further quantify the impact of data heterogeneity. \MethodName\ has a maximum standard deviation of \rev{1.66} across all models and datasets, compared to 8.19, 5.75, 10.11, and 11.15 for SplitFed, PiPar, SCAFFOLD, and SplitGP, respectively\rev{; in all cases the maximum is achieved for classification, with Group~B's mIoU standard deviations being smaller}. \MethodName\ reduces this standard deviation by \rev{71.13\%} compared to the most stable baseline, indicating that \MethodName\ is less impacted by non-IID data.

As $\alpha$ decreases from $1$ to $0.1$, the accuracy gap to the strongest SFL baseline widens while \MethodName\ stays close to its IID accuracy across all four groups: \rev{the accuracy of Groups~A, C, and D, and mIoU for Group~B drops by at most 3.27 percentage points from the IID to the severely non-IID setting, and at $\alpha=0.1$ it still leads the strongest baseline by 2.28--12.50 percentage points.}

\begin{tcolorbox}[
    width=0.49\textwidth,
    colframe=black!50!black,
    colback=gray!8,boxrule=0.5pt,
    left=2pt,right=2pt,top=2pt,bottom=2pt]
\textbf{Observation 4:}
\MethodName\ achieves higher accuracy \rev{for image and text classification and higher mIoU for semantic segmentation} as well as lower variance than SFL baselines across non-IID degrees, \rev{with increasing gain over the baselines} as data becomes more heterogeneous.
\end{tcolorbox}

\subsection{Ablation Studies}
\label{subsec:ablation}

We conduct ablation studies to isolate the contribution of each design choice in \MethodName. They demonstrate the benefits of activation consolidation for improving accuracy\rev{, and validate the auxiliary network depth and split point choices}.

\subsubsection{Impact of Activation Consolidation}

To validate the effectiveness of activation consolidation, an ablation study that compares the model accuracy \rev{for Groups~A, C and D, and mIoU for Group~B} of \MethodName\ with and without activation consolidation is considered.

In \MethodName\ without activation consolidation, device block training remains unchanged. However, once the device block converges, each device sends its activations to the server, which maintains $K$ separate activation sets $\mathcal{A}_k,k\in[K]$, each corresponding to a specific device as in an SFL system. After collecting all activations, the server trains $K$ individual server blocks on their respective activation sets and periodically aggregates them into a global server block, following the approach used in SFL baselines.

As shown in Figure~\ref{fig:act_merge}, activation consolidation improves the final accuracy of \MethodName\ across all groups, mitigating the non-IID challenge, with the largest gains under the most severe class skew. \rev{Group~A gains 6.04 to 13.20 percentage points.} \rev{Group~B gains 3.65 percentage points in mIoU on Cityscapes and 4.73 percentage points on LIP.} \rev{Group~C gains 6.06 to 12.05 percentage points.} \rev{Group~D gains 3.55 percentage points on AG News and 6.86 percentage points on SST-2.}

%

\begin{figure}[tp]
    \centering
    \begin{tikzpicture}[scale=1]
    \begin{axis}[
        xbar,
        width=0.65\linewidth,
        height=8.5cm,
        xmin=0, xmax=100,
        xlabel={\rev{Accuracy / mIoU (\%)}},
        xlabel style={font=\footnotesize},
        symbolic y coords={
            mobci,vggci,mobti,vggti,
            fcnvoc,fcnplan,
            stci,vtci,stti,vtti,
            daag,dsst
        },
        ytick=data,
        y dir=reverse,
        yticklabels={
            MobileNet-L on CIFAR-10,
            VGG-11 on CIFAR-10,
            MobileNet-L on Tiny ImageNet,
            VGG-11 on Tiny ImageNet,
            FCN-ResNet18 on Cityscapes,
            FCN-ResNet18 on LIP,
            Swin-T on CIFAR-10,
            ViT-S on CIFAR-10,
            Swin-T on Tiny ImageNet,
            ViT-S on Tiny ImageNet,
            DistilBERT on AG News,
            DistilBERT on SST-2
        },
        bar width=4pt,
        enlarge y limits=0.04,
        legend style={at={(-0.8,1.02)}, anchor=south west, legend columns=2, font=\footnotesize, draw=none, /tikz/every even column/.append style={column sep=0.4cm}},
        xticklabel style={font=\footnotesize},
        yticklabel style={align=right, font=\footnotesize}
    ]
        \addplot[fill=Red,bar shift=-0.08cm] coordinates {
            (57.36,mobci) (74.45,vggci)
            (11.44,mobti) (28.43,vggti)
            (24.77,fcnvoc) (19.21,fcnplan)
            (55.59,stci) (50.66,vtci)
            (19.56,stti) (15.17,vtti)
            (84.48,daag) (74.52,dsst)
        };
        \addlegendentry{without consolidation}
        \addplot[fill=Green,bar shift=0.08cm] coordinates {
            (68.60,mobci) (80.49,vggci)
            (24.60,mobti) (41.63,vggti)
            (28.42,fcnvoc) (23.94,fcnplan)
            (67.64,stci) (61.04,vtci)
            (29.18,stti) (21.23,vtti)
            (88.03,daag) (81.38,dsst)
        };
        \addlegendentry{with consolidation}
    \end{axis}
    \end{tikzpicture}
    \caption{\rev{Model accuracy for image and text classification and mean Intersection over Union (mIoU; higher is better for both) for semantic segmentation using \MethodName\ with and without activation consolidation.}}
    \label{fig:act_merge}
\end{figure}

\begin{tcolorbox}[
    width=0.49\textwidth,
    colframe=black!50!black,
    colback=gray!8,boxrule=0.5pt,
    left=2pt,right=2pt,top=2pt,bottom=2pt]
\textbf{Observation 5:}
Activation consolidation mitigates non-IID challenges and improves model accuracy \rev{(and mIoU for relevant tasks over different architectures)} in \MethodName.
\end{tcolorbox}

\subsubsection{\rev{Impact of Auxiliary Network Design}}
\label{subsubsec:aux_ablation}

\rev{To validate the best configuration of the auxiliary network and its impact on model accuracy, we conduct an ablation study varying the number of layers in the auxiliary network. Three configurations are evaluated:}
\rev{\begin{itemize}[leftmargin=*]
    \item \textbf{1-layer (FC only):} A single fully connected layer that maps the device block output to the classification loss.
    \item \textbf{2-layer (default):} The standard design used in \MethodName, consisting of a convolutional/linear layer at half dimension followed by a fully connected classification head.
    \item \textbf{3-layer:} An additional intermediate layer is inserted between the first layer and the classification head.
\end{itemize}}

\rev{The depths $\{1,2,3\}$ span the full range of feasible designs: the 1-layer configuration is the shallowest auxiliary network that can map device-block features to a loss, and the 3-layer configuration is the deepest that fits the 4\,GB unified-memory budget on every backbone (a fourth half-dimension layer exceeds it for the larger transformers at the operational batch size). For each configuration we report the device block accuracy (with the auxiliary network) and the final model accuracy (after server block training) to quantify how well the auxiliary task transfers to the server block.}

\begin{table*}[tp]
    \centering
    \caption{\rev{Impact of the auxiliary network depth on \MethodName\ for MobileNet-L on CIFAR-10 and DistilBERT on AG News. Dev. Acc.: validation accuracy achieved at the end of device block training; Final Acc.: test accuracy after server block training; Time: total training time (hours); Comm.: per-device device--server communication volume (GB); $N^{(d)}$: number of device epochs until convergence. The highlighted row is the default configuration selected for \MethodName.}}
    \label{tab:aux_ablation}
    \begin{tabular}{c|ccccc|ccccc}
        \hline
        \colorrow
        & \multicolumn{5}{c|}{\textbf{\rev{MobileNet-L (CIFAR-10)}}} & \multicolumn{5}{c}{\textbf{\rev{DistilBERT (AG News)}}} \\
        \hhline{~|----------}
        \colorrow
        \multirow{-2}{*}{\textbf{\rev{Config}}}
        & \footnotesize \rev{Dev. Acc.} & \footnotesize \rev{Final Acc.} & \footnotesize \rev{Time} & \footnotesize \rev{Comm.} & \footnotesize \rev{$N^{(d)}$}
        & \footnotesize \rev{Dev. Acc.} & \footnotesize \rev{Final Acc.} & \footnotesize \rev{Time} & \footnotesize \rev{Comm.} & \footnotesize \rev{$N^{(d)}$} \\
        \hline
        \rev{1-layer} & \rev{54.39\%} & \rev{64.42\%} & \rev{0.67\,h} & \rev{0.067\,GB} & \rev{38} & \rev{83.04\%} & \rev{85.43\%} & \rev{1.09\,h}  & \rev{3.42\,GB}  & \rev{25} \\
        \highlightrow
        \rev{2-layer} & \rev{57.21\%} & \rev{67.34\%} & \rev{0.71\,h} & \rev{0.067\,GB} & \rev{55}  & \rev{85.11\%} & \rev{88.03\%} & \rev{3.55\,h}  & \rev{8.20\,GB}  & \rev{68} \\
        \rev{3-layer} & \rev{57.96\%} & \rev{68.49\%} & \rev{1.08\,h} & \rev{0.102\,GB} & \rev{123}  & \rev{87.79\%} & \rev{90.84\%} & \rev{6.48\,h}  & \rev{10.18\,GB} & \rev{84} \\
        \hline
    \end{tabular}%
\end{table*}

\rev{Table~\ref{tab:aux_ablation} shows the same trade-off on both architectures: deeper auxiliary networks improve accuracy at a growing cost in device epochs, training time, and communication. The 1-layer (FC-only) design is the cheapest but lowers final accuracy by 2.92 and 2.60 percentage points relative to the 2-layer default, because without the intermediate layer the auxiliary objective lacks the inductive bias needed for the device block to learn features that transfer to the server block (Section~\ref{subsubsec:aux}). The 3-layer design is the most accurate but costs up to 2.2$\times$ the device epochs and 1.5$\times$--1.8$\times$ the training time. The 2-layer configuration sits at the Pareto knee, capturing most of the attainable accuracy at roughly half the cost of the deepest, and is therefore selected as the default across both architectures; the 3-layer design is an alternative when accuracy is prioritized over cost, specifically for Transformers.}

\begin{tcolorbox}[
    width=0.49\textwidth,
    colframe=black!50!black,
    colback=gray!8,boxrule=0.5pt,
    left=2pt,right=2pt,top=2pt,bottom=2pt]
\rev{\textbf{Observation 6:}
The 2-layer auxiliary network design of \MethodName\ balances between feature generalization and computational efficiency across both CNN and Transformer architectures.}
\end{tcolorbox}

\subsubsection{\rev{Impact of Split Point on Final Accuracy}}
\label{subsubsec:split_ablation}

\rev{Section~\ref{subsubsec:uit} highlighted that $p=1$ is the optimal split point for UIT since it simultaneously minimizes on-device computation and device-server communication. To verify that $p=1$ also achieves a competitive final accuracy, we explore $p \in \{1, 2, 3, 5\}$ on a CNN (MobileNet-L on CIFAR-10) and a Transformer (DistilBERT on AG News); $p>5$ exceeds the 4\,GB device memory budget and is therefore not considered. For each split point we report the final model accuracy, total training time, per-device communication volume, and the device epochs $N^{(d)}$ to convergence.}

\begin{table*}[tp]
    \centering
    \caption{\rev{Impact of split point $p$ on final model accuracy (Acc.), training time in hours (Time), per-device communication volume in GB (Comm.), and number of device epochs until convergence ($N^{(d)})$ for MobileNet-L on CIFAR-10 and DistilBERT on AG News. The highlighted row ($p=1$) is the default configuration selected for \MethodName.}}
    \label{tab:split_ablation}
    \begin{tabular}{c|cccc|cccc}
        \hline
        \colorrow
        & \multicolumn{4}{c|}{\textbf{\rev{MobileNet-L (CIFAR-10)}}} & \multicolumn{4}{c}{\textbf{\rev{DistilBERT (AG News)}}} \\
        \hhline{~|--------}
        \colorrow
        \multirow{-2}{*}{\rev{$p$}}
        & \footnotesize \rev{Acc.} & \footnotesize \rev{Time} & \footnotesize \rev{Comm.} & \footnotesize \rev{$N^{(d)}$}
        & \footnotesize \rev{Acc.} & \footnotesize \rev{Time} & \footnotesize \rev{Comm.} & \footnotesize \rev{$N^{(d)}$} \\
        \hline
        \highlightrow
        \rev{1} & \rev{67.34\%} & \rev{0.71\,h} & \rev{0.067\,GB} & \rev{55} & \rev{88.03\%} & \rev{3.55\,h}  & \rev{8.20\,GB}  & \rev{68} \\
        \rev{2} & \rev{64.62\%} & \rev{0.80\,h} & \rev{0.070\,GB} & \rev{82} & \rev{86.19\%} & \rev{6.51\,h}  & \rev{9.88\,GB}  & \rev{80} \\
        \rev{3} & \rev{59.75\%} & \rev{0.65\,h} & \rev{0.028\,GB} & \rev{45} & \rev{84.66\%} & \rev{7.85\,h}  & \rev{9.27\,GB}  & \rev{71} \\
        \rev{5} & \rev{58.91\%} & \rev{1.00\,h} & \rev{0.035\,GB} & \rev{93} & \rev{81.25\%} & \rev{11.85\,h} & \rev{9.63\,GB}  & \rev{68} \\
        \hline
    \end{tabular}%
\end{table*}

\rev{Table~\ref{tab:split_ablation} shows that $p=1$ achieves the highest final accuracy for MobileNet-L and DistilBERT and accuracy falls monotonically as $p$ grows. Increasing $p$ shrinks the server block, which reduces the capacity for representational learning that follows activation consolidation, while adding more parameters to the local auxiliary objective. Both effects amplify the impact of non-IID data on devices.}

\rev{The training time and communication columns confirm $p=1$ is also at or near the cost optimum. For DistilBERT $p=1$ requires the least training time (3.55\,h) and communication (8.20\,GB) while achieving the best accuracy; $p\in\{2,3,5\}$ increases training time by 1.83$\times$--3.34$\times$ at lower accuracy. For MobileNet-L, $p=3$ is an outlier that converges in only $N^{(d)}=45$ epochs, so the training time and communication volume dip below $p=1$, but at a 7.59-percentage-point accuracy loss. These results confirm the analysis of Section~\ref{subsubsec:uit}: $p=1$ is the cost-minimizing split point estimated by Equation~\labelcref{eq:comm} and achieves the highest accuracy. Since the device block has only one layer, the server block is large, and consequently, the server absorbs the majority of the network's representational load, resulting in an accuracy gain.}

\begin{tcolorbox}[
    width=0.49\textwidth,
    colframe=black!50!black,
    colback=gray!8,boxrule=0.5pt,
    left=2pt,right=2pt,top=2pt,bottom=2pt]
\rev{\textbf{Observation 7:}
The split point at which the device does the least training, offloading the rest of the network to the server, achieves the highest final accuracy while minimizing both training time and communication overhead.}
\end{tcolorbox}

\section{Related Work}
\label{sec:rw}
In this section, we review existing methods that reduce communication costs and improve model accuracy in SFL, analyze their limitations, and contrast them with \MethodName.

\rev{\textbf{Split learning and edge collaborative training}.
In split learning (SL)~\cite{sl,slh}, a neural network is partitioned between a device and a server, and the intermediate activations and gradients are exchanged during each training iteration. The server computation of the network in SL reduces on-device computation. However, classic SL requires each device network to have a server counterpart and requires sequential communication for every training iteration, both of which limits scalability. SplitFed~\cite{splitfed} extends SL to the multi-device FL setting by combining SL's model partitioning with FedAvg-style aggregation across devices. Memory constraints in heterogeneous FL is further addressed by adaptive model splitting, enabling resource-limited devices with diverse memory capacities to participate in collaborative training~\cite{tian2024breaking}. A recent survey systematically evaluates SL paradigms across label sharing, model aggregation, and cut layer selection strategies~\cite{slsurvey2025}. The broader edge intelligence paradigm~\cite{edgeintelligence} distributes AI workloads across heterogeneous edge devices and cloud servers; within this paradigm, edge collaborative training approaches address resource constraints by using techniques such as pipeline parallelism~\cite{pipar,parallelsfl} and efficient parallel split learning over resource-constrained wireless networks~\cite{lin2024efficient}. Beyond terrestrial deployments, semantic communication-based computation offloading in satellite-borne edge cloud networks~\cite{zheng2024semantic} demonstrates the broader applicability of edge-cloud collaborative inference for diverse network topologies. }

\rev{\MethodName\ differs from both SL and edge collaborative training: while they retain iterative activation and gradient exchanges between devices and servers, \MethodName\ eliminates this dependency entirely through unidirectional inter-block training. The device block is first trained to convergence using a lightweight auxiliary network, and activations are then transferred only once, removing the per-iteration communication bottleneck.
}

\textbf{Reducing communication overhead arising from intermediate results}.
Existing research has explored four strategies to reduce communication overheads arising from the exchange of activations and gradients in SFL systems:

\textbf{(1) Generating gradients locally:} 
Gradients are computed on devices to eliminate server-side gradient transmission, reducing SFL communication by half~\cite{han2021accelerating}. FedGKT~\cite{fedgkt} deploys local loss functions on devices and distills knowledge between the local and global models. EcoFed~\cite{ecofed} reduces communication using a pre-trained device block but relies on public data, limiting its use in the real world as public data may not correspond to private device data.
\textit{However, activations still need to be transferred, resulting in SFL communication higher than its FL counterpart.}

\textbf{(2) Overlapping communication and computation:} PiPar~\cite{pipar} utilizes pipeline parallelism to overlap computation with activation transfers, thereby reducing idle time. This improves training efficiency by computing and communicating concurrently, \textit{but does not reduce the total communication volume or frequency.}

\textbf{(3) Clustering devices:} 
FedAdapt~\cite{fedadapt} and ParallelSFL~\cite{parallelsfl} group devices with similar compute and bandwidth capacities, optimizing communication within clusters.
\textit{However, the trade-off between on-device computation and device-server communication is unresolved for individual devices.}

\textbf{(4) Compressing activation/gradients:}
In split learning, activations and gradients are compressed into lower-dimensional representations~\cite{bottleneck,c3sl}. However, these methods are not designed for FL with multiple devices and impact model accuracy. BaF~\cite{b2f} selectively compresses activations and applies a predictive reconstruction method, while randomized top-k sparsification~\cite{topksparse} aims to retain accuracy.
\textit{While these methods reduce activation size, they are less efficient than FL that does not require activation transfers.}

Other FL systems reduce communication by training a sub-model instead of an entire model on a device~\cite{10.1145/3447993.3483278,10.1145/3372224.3419188}, using hyperdimensional computing~\cite{10.1145/3495243.3558757}, \rev{or applying compressive sensing to model parameters~\cite{fedcs}}. However, these assume high on-device computational capabilities and are not suited for SFL.

Existing methods fail to reduce the communication volume or frequency in SFL systems to FL levels because they still rely on iterative activation transfers. 
In contrast, \MethodName\ reduces per-iteration activation transfers to a one-shot activation exchange, thereby significantly reducing both communication volume and frequency.

\textbf{Improving model accuracy under non-IID data conditions}. 
Existing systems for improving FL model accuracy that train with non-IID data either modify local loss functions or global aggregation algorithms. FedAvgM~\cite{fedavgm} adjusts the way models are updated to make convergence more stable, reducing the impact of variable local updates. FedProx~\cite{fedprox} introduces a regularization term that restricts local model updates, preventing them from deviating significantly from the global model. SCAFFOLD~\cite{scaffold} addresses variations in local model updates by maintaining correction terms that adjust for differences in gradient updates across devices, ensuring that all devices contribute more uniformly to the global model. SplitGP~\cite{splitgp} combines the gradients produced by local loss on devices and by global loss on the server to balance generalization and personalization, mitigating the non-IID issue. \rev{Serial pipeline training combined with global knowledge regularization is explored to address data heterogeneity across devices~\cite{luo2023improving}. FairFed~\cite{fairfed} improves contribution evaluation fairness in FL under non-IID data using cooperative Shapley value.}

While these systems mitigate the non-IID issue of the original data to some extent, they do not specifically address activation heterogeneity in SFL. Instead, each device block is trained separately on local data, and the server block continues to receive activations that reflect non-IID distributions. Thus, the server block learns inconsistent feature representations from device-specific activations, reducing generalization and degrading model accuracy.

\MethodName\ addresses the challenge by consolidating activations into a single dataset for server block training. As the consolidated dataset is more homogeneous than each device's activations, \MethodName\ improves overall model accuracy without requiring specialized aggregation algorithms.

\section{Conclusion and Discussion}
\label{sec:concl}
Three challenges limit the use of a Split Federated Learning (SFL) system in the real-world. Firstly, it has high communication overheads arising from frequent exchanges of intermediate activations and gradients between devices and the server. Secondly, the trade-off between on-device computation and device-server communication when selecting the split point leads to less optimal performance. Thirdly, data heterogeneity impacts model accuracy. 

We propose \MethodName\, a new collaborative training system that goes beyond SFL to address these challenges by introducing three key innovations: (1) unidirectional inter-block training, a novel technique to sequentially train device and server layers instead of using end-to-end backpropagation in SFL systems. This eliminates the computation and communication trade-off and achieves both optimal on-device computation and device-server communication at the same split point; (2) a lightweight auxiliary network generation method, which decouples device and server training, eliminating iterative transfers of intermediate results in SFL systems and significantly reducing communication volume and frequency; and (3) an activation consolidation method, which trains a single server block on a unified activation set generated by all devices using converged device layers, rather than training on device-specific, non-IID activations as in SFL, thereby mitigating the impact of data heterogeneity. 

\MethodName\ is evaluated against state-of-the-art SFL baselines using CNN and Transformer models. Results show that \MethodName\ improves accuracy by up to \rev{11.70 percentage points} while training up to \rev{18.6$\times$} faster, incurring up to \rev{911$\times$} lower device-server communication overhead, and up to \rev{14.5$\times$} lower on-device computation. Furthermore, \MethodName\ reduces the standard deviation of accuracy by \rev{71.13\%} for various non-IID degrees compared to other baselines, indicating that it is less impacted by non-IID data.

While \MethodName\ effectively addresses key challenges in split federated learning, there are open avenues for improvement as considered below.

\textit{Privacy considerations.} Like most SFL systems, \MethodName\ requires devices to upload intermediate activations to the server. Although activations are transferred only once in \MethodName, prior research~\cite{mahendran2015understanding} has shown that intermediate features may be vulnerable to inversion attacks, particularly in vision models. \MethodName\ is compatible with several privacy-preserving techniques, including differentially private training (e.g., DP-SGD~\cite{dp}), homomorphic encryption~\cite{cheon2017homomorphic} for encrypted activation transfer, and certified defenses such as PixelDP~\cite{pixeldp} that target privacy leakage from intermediate representations. These techniques can be integrated in the context of one-shot activation transfer to extend \MethodName\ for privacy-sensitive applications.

\rev{\textit{Task diversity.} \MethodName\ has been evaluated on a range of tasks, including image classification (using CNN and Vision Transformer backbones), text classification (using a Transformer language model), and semantic segmentation (using a CNN backbone). The evaluation of \MethodName\ will be extended to include other task families, such as object detection and generative modeling.}

\textit{Hardware constraints.} 
Although \MethodName\ reduces on-device computation compared to existing FL and SFL systems, it still assumes devices capable of executing a shallow device block and a lightweight auxiliary network. Severely resource-constrained platforms, such as microcontrollers or bare-metal sensors, are outside the current scope of this work. Extending the system to support such ultra-low-end devices remains a direction for future investigation.


%


\ifCLASSOPTIONcompsoc
  \section*{Acknowledgments}
\else
  \section*{Acknowledgment}
\fi
This research is funded by Rakuten Mobile, Inc., Japan, and supported by the UKRI grant EP/Y028813/1.

\ifCLASSOPTIONcaptionsoff
  \newpage
\fi



\bibliographystyle{IEEEtran}
\bibliography{ref.bib}

%
\vskip -1.5\baselineskip plus -1fil
\begin{IEEEbiography}
[{\includegraphics[width=1in,height=1.25in,clip,keepaspectratio]{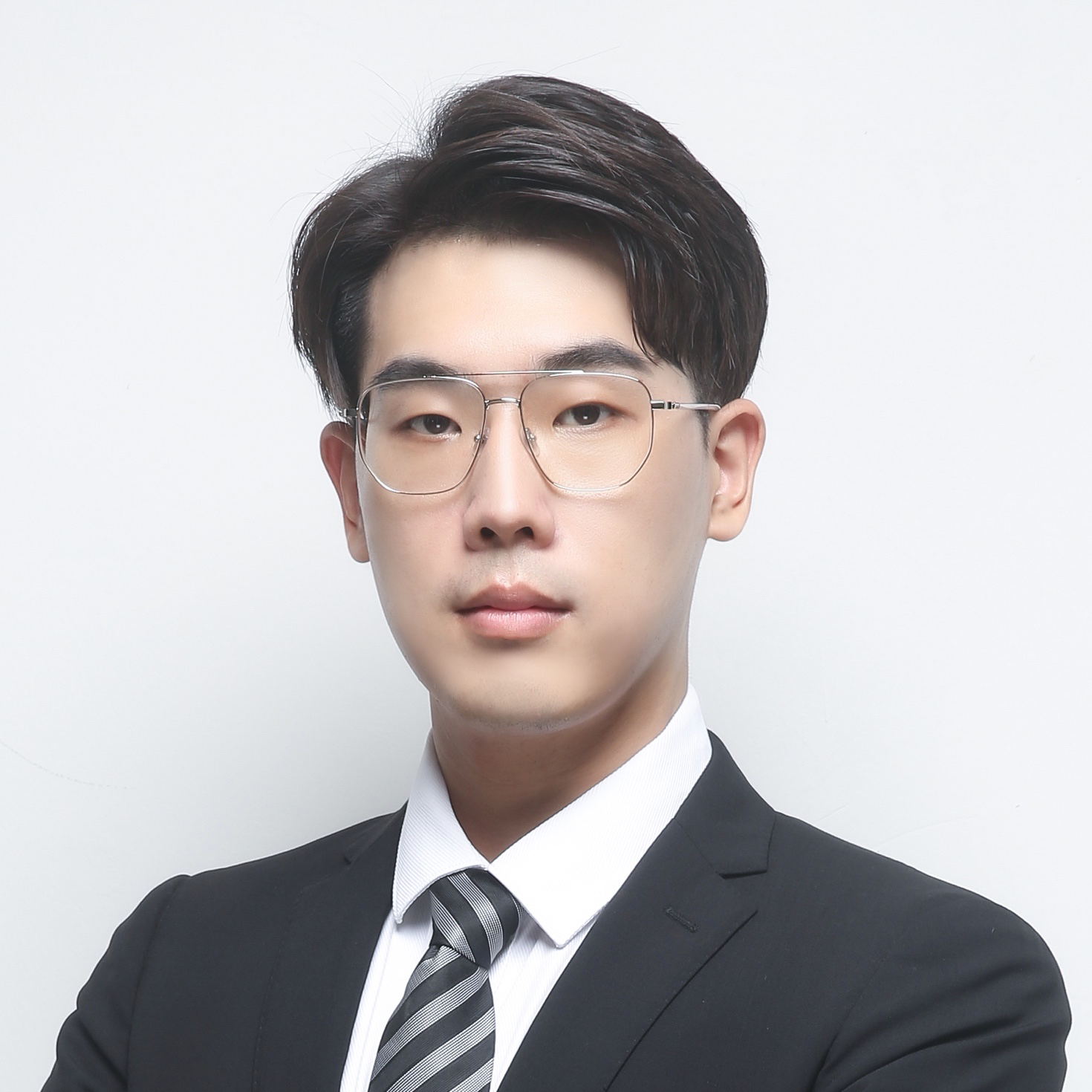}}]{Zihan Zhang}
is a PhD student at University of St Andrews. He was previously a researcher at Noah’s Ark Lab, Huawei. Prior to that, he obtained a Master's degree in Data Science with distinction from University of Edinburgh and a Bachelor degree in Mathematics from Shandong University. His interests include collaborative machine learning and cloud/edge computing.
\end{IEEEbiography}
\vskip -2.5\baselineskip plus -1fil

\begin{IEEEbiography}[{\includegraphics[width=1in,height=1.25in,clip,keepaspectratio]{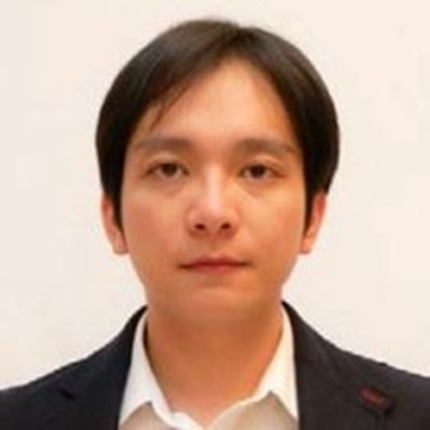}}]{Leon Wong}
is the Research Collaboration and Engineering Lead for Autonomous Networks Research \& Innovation in Rakuten Mobile, Inc. He is currently serving as chairman of ITU-T Focus Group of Autonomous Networks (FG-AN), established under ITU-T Study Group 13 - Future networks and emerging network technologies. He is also the co-chair of FG-AN Ad hoc group for Japan's Telecommunication Technology Committee (TTC).
\end{IEEEbiography}
\vskip -2.5\baselineskip plus -1fil

\begin{IEEEbiography}[{\includegraphics[width=1in,height=1.25in,clip,keepaspectratio]{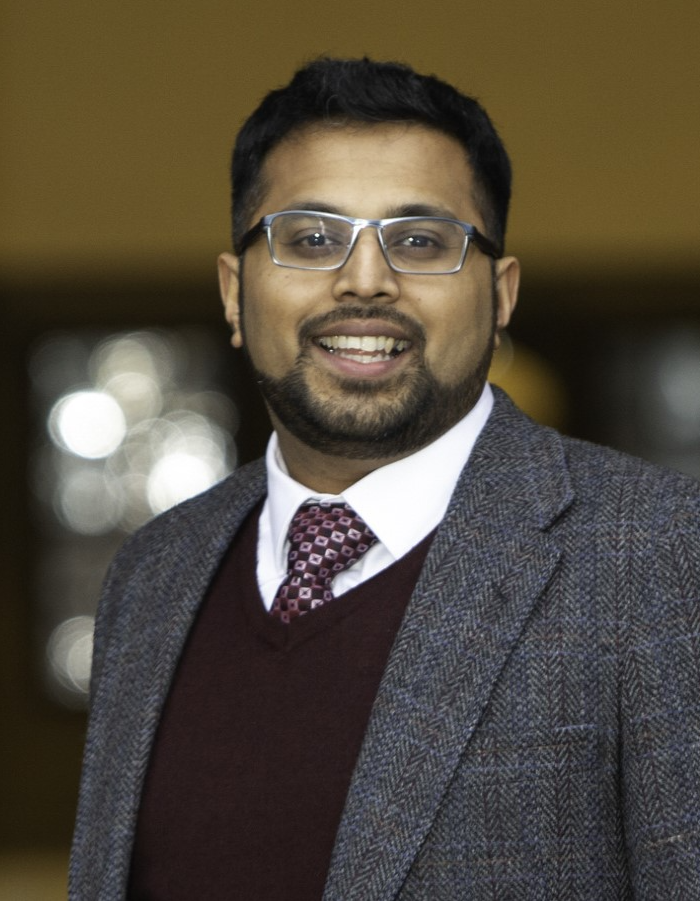}}]{Blesson Varghese}
is a Professor in Computer Science at the University of St Andrews, UK, and the Director of the Edge Computing Hub. He leads research on edge computing for AI in the UK's flagship National Edge AI Hub. He is a previous Royal Society Short Industry Fellow. His current interests are in distributed machine learning systems. More information is available from \url{www.blessonv.com}.
\end{IEEEbiography}







\end{document}